\documentclass[sn-mathphys]{sn-jnl}

\usepackage{adjustbox}

\usepackage[ruled, vlined, nofillcomment, linesnumbered]{algorithm2e}

\usepackage{amsmath,amssymb,amsfonts,scrextend}

\usepackage{tikz}
\usepackage{pgfplots}
\pgfplotsset{compat=1.18}
\usetikzlibrary{positioning}
\usepackage{pgfplotstable}
\usepackage{booktabs}
\usepackage[numbers,sort&compress]{natbib}

\newcommand{\set}[1]{\left\{#1\right\}}

\newcommand{\fpr}[1]{\mathopen{}\left(#1\right)}

\newcommand{\abs}[1]{{\left|#1\right|}}

\newcommand{\enset}[2]{\left\{#1 ,\ldots , #2\right\}}

\newcommand{\np}{\textbf{NP}}
\newcommand{\poly}{\textbf{P}}

\newcommand{\lab}{\mathit{lab}}

\newcommand{\fm}[1]{{\mathcal{#1}}}

\DeclareRobustCommand{\dispfunc}[2]{%
    \ensuremath{%
        \ifthenelse{\equal{#2}{}}%
            {\mathit{#1}}%
            {\mathit{#1}\fpr{#2}}}}

\newcommand{\bigO}[1]{\dispfunc{\mathcal{O}}{#1}}

\newcommand{\prbld}{\textsc{LD}\xspace}
\newcommand{\prbldor}{\textsc{LDor}\xspace}
\newcommand{\prbldand}{\textsc{LDand}\xspace}
\newcommand{\prbldalpha}{\textsc{LD-$\alpha$}\xspace}

\newcommand{\alggreedyand}{\textsc{GreedyAnd}\xspace}
\newcommand{\alggreedyor}{\textsc{GreedyOr}\xspace}
\newcommand{\alggreedyandalpha}{\textsc{GreedyAnd-$\alpha$}\xspace}
\newcommand{\alggreedyoralpha}{\textsc{GreedyOr-$\alpha$}\xspace}
\newcommand{\algbase}{\textsc{Dense}\xspace}

\newcommand{\indand}[1]{\dispfunc{f_{\mathit{AND}}}{#1}}
\newcommand{\indor}[1]{\dispfunc{f_{\mathit{OR}}}{#1}}

\newcommand{\hull}[1]{\dispfunc{hull}{#1}}

\theoremstyle{thmstyleone}
\newtheorem{theorem}{Theorem}
\newtheorem{proposition}{Proposition}
\newtheorem{problem}{Problem}
\newtheorem{lemma}{Lemma}
\newtheorem{corollary}{Corollary}

\theoremstyle{thmstyletwo}

\theoremstyle{thmstylethree}
\newtheorem{definition}{Definition}

\definecolor{yafcolor1}{rgb}{0.4, 0.165, 0.553}
\definecolor{yafcolor2}{rgb}{0.949, 0.482, 0.216}
\definecolor{yafcolor3}{rgb}{0.47, 0.549, 0.306}
\definecolor{yafcolor4}{rgb}{0.925, 0.165, 0.224}
\definecolor{yafcolor5}{rgb}{0.141, 0.345, 0.643}
\definecolor{yafcolor6}{rgb}{0.965, 0.933, 0.267}
\definecolor{yafcolor7}{rgb}{0.627, 0.118, 0.165}
\definecolor{yafcolor8}{rgb}{0.878, 0.475, 0.686}

\begin{document}

\title{Dense subgraphs induced by edge labels}

\author*[1]{\fnm{Iiro} \sur{Kumpulainen}}\email{iiro.kumpulainen@helsinki.fi}

\author*[1]{\fnm{Nikolaj} \sur{Tatti}}\email{nikolaj.tatti@helsinki.fi}

\affil[1]{HIIT, University of Helsinki, Helsinki}

\abstract{
Finding densely connected groups of nodes in networks is a widely used tool for analysis in graph mining. A popular choice for finding such groups is to find subgraphs with a high average degree.
While useful, interpreting such subgraphs may be difficult.
On the other hand, many real-world networks have additional information, and we are specifically interested in networks with labels on edges.
In this paper, we study finding sets of labels that induce dense subgraphs.
We consider two notions of density: average degree and the number of edges minus the number of nodes weighted by a parameter $\alpha$.
There are many ways to induce a subgraph from a set of labels, and we study two cases:
First, we study conjunctive-induced dense subgraphs, where the subgraph edges need to have all labels. Secondly, we study disjunctive-induced dense subgraphs, where the subgraph edges need to have at least one label.
We show that both problems are \np-hard. Because of the hardness, we resort to greedy heuristics. We show that we can implement the greedy
search efficiently: the respective running times for finding conjunctive-induced and disjunctive-induced dense subgraphs are in $\bigO{p \log k}$ and $\bigO{p \log^2 k}$, where $p$ is the number of edge-label pairs and $k$ is the number of labels.
Our experimental evaluation demonstrates that we can find the ground truth in synthetic graphs and that we can find interpretable subgraphs from real-world networks.
}


\maketitle

\section{Introduction}

Finding dense subgraphs in networks is a common tool for analyzing networks with potential applications in diverse domains, such as bioinformatics~\citep{fratkin2006motifcut,langston2005combinatorial},
finance~\citep{du2009migration},
social media~\citep{angel2014dense}, or web graph analysis~\citep{fratkin2006motifcut}.

While useful on their own, analyzing dense subgraphs without any additional explanation may be difficult for domain experts and consequently may limit its usability. 

Fortunately, it is often the case that the network has additional information such as labels associated with nodes and/or edges. For example, in social networks, users may have tags describing themselves. In networks arising from communication, for example, by email or Twitter, the tags associated with the edge can be the tags associated or extracted with the message.

Using the available label information to provide explainable dense subgraphs may ease the burden of domain experts when, for example, studying social networks.
In this paper, we consider finding dense subgraphs in networks with labeled edges. More formally, we are looking for a label set that \emph{induces} a dense subgraph. As a measure of density, a subgraph $(W, F)$ with nodes $W$ and edges $F$ will use $\abs{F} / \abs{W}$, the ratio of edges over the nodes, a popular choice for measuring the density of a subgraph.

We consider two cases: conjunctive-induced and disjunctive-induced dense subgraphs. In the former, the induced subgraph consists of all the edges that have the given label set. In the latter, the induced subgraph consists of all the edges that have at least one label common with the label set. We give an example of both cases in Figure~\ref{fig:example}.

\begin{figure}[hb!]
\hspace*{\fill}
\begin{tikzpicture}
\begin{scope}[yshift=-0.2cm]

\node[circle, inner sep=2pt] (k1) at (0*72+18:0.5) {};
\draw[thick,fill=white, draw=black] (k1) circle (3pt);

\node[circle, inner sep=2pt] (k2) at (1*72+18:0.5) {};
\draw[thick,fill=white, draw=black] (k2) circle (3pt);

\node[circle, inner sep=2pt] (k3) at (2*72+18:0.5) {};
\draw[thick,fill=white, draw=black] (k3) circle (3pt);

\node[circle, inner sep=2pt] (k4) at (3*72+18:0.5) {};
\draw[thick,fill=white, draw=black] (k4) circle (3pt);

\node[circle, inner sep=2pt] (k5) at (4*72+18:0.5) {};
\draw[thick,fill=white, draw=black] (k5) circle (3pt);

\end{scope}

\node[circle, inner sep=2pt, right=0.3cm of k1] (u1) {};
\draw[thick,fill=white, draw=black] (u1) circle (3pt);

\node[circle, inner sep=2pt, above left=0.1 and 0.3 of k3] (v1) {};
\draw[thick,fill=white, draw=black] (v1) circle (3pt);

\node[circle, inner sep=2pt, below left=0.1 and 0.3 of k3] (v2) {};
\draw[thick,fill=white, draw=black] (v2) circle (3pt);

\foreach \x [count=\xi from 2] in {1,...,4}
\foreach \y in {\xi,...,4} {
\draw[ultra thick, yafcolor1!50!white] (k\x) edge (k\y);
\draw[ultra thick, dashed, yafcolor2!50!white] (k\x) edge (k\y);
}
\draw[ultra thick, yafcolor1!50!white] (k5) edge (k1);
\draw[ultra thick, dashed, yafcolor2!50!white] (k5) edge (k1);

\draw[ultra thick, yafcolor2!50!white] (k1) edge (u1);
\draw[ultra thick, yafcolor1!50!white] (k3) edge (v1);
\draw[ultra thick, yafcolor1!50!white] (v1) edge (v2);
\draw[ultra thick, yafcolor1!50!white] (v2) edge (k3);

\node[circle, inner sep=2pt] (l1) at (1.5, 0) {};
\draw[thick,fill=yafcolor2!50!white, draw=yafcolor2] (l1) circle (3pt);
\node[anchor=west, inner sep=6pt] at (l1) {$= \ell_1$};
\node[circle, inner sep=2pt] (l1) at (1.5, -0.4) {};
\draw[thick,fill=yafcolor1!50!white, draw=yafcolor1] (l1) circle (3pt);
\node[anchor=west, inner sep=6pt] at (l1) {$= \ell_2$};

\end{tikzpicture}\hfill
\begin{tikzpicture}

\begin{scope}[yshift=-0.2cm]

\node[circle, inner sep=2pt] (k1) at (0*72+18:0.5) {};
\draw[thick,fill=white, draw=black] (k1) circle (3pt);

\node[circle, inner sep=2pt] (k2) at (1*72+18:0.5) {};
\draw[thick,fill=white, draw=black] (k2) circle (3pt);

\node[circle, inner sep=2pt] (k3) at (2*72+18:0.5) {};
\draw[thick,fill=white, draw=black] (k3) circle (3pt);

\node[circle, inner sep=2pt] (k4) at (3*72+18:0.5) {};
\draw[thick,fill=white, draw=black] (k4) circle (3pt);

\node[circle, inner sep=2pt] (k5) at (4*72+18:0.5) {};
\draw[thick,fill=white, draw=black] (k5) circle (3pt);

\end{scope}

\node[circle, inner sep=2pt, above right=0.3 and 0.2 of k1] (u1) {};
\draw[thick,fill=white, draw=black] (u1) circle (3pt);

\node[circle, inner sep=2pt, right=0.3 of k1] (u2) {};
\draw[thick,fill=white, draw=black] (u2) circle (3pt);

\node[circle, inner sep=2pt, below right=0.3 and 0.2 of k1] (u3) {};
\draw[thick,fill=white, draw=black] (u3) circle (3pt);

\draw[ultra thick, yafcolor1!50!white] (k1) edge (k2);
\draw[ultra thick, dashed, yafcolor2!50!white] (k1) edge (k2);
\draw[ultra thick, yafcolor1!50!white] (k2) edge (k3);
\draw[ultra thick, yafcolor2!50!white] (k3) edge (k4);
\draw[ultra thick, yafcolor2!50!white] (k4) edge (k2);
\draw[ultra thick, yafcolor2!50!white] (k4) edge (k5);
\draw[ultra thick, yafcolor1!50!white] (k4) edge (k1);
\draw[ultra thick, yafcolor1!50!white] (k2) edge (k5);
\draw[ultra thick, yafcolor3!50!white] (k5) edge (k1);

\draw[ultra thick, yafcolor3!50!white] (k1) edge (u1);
\draw[ultra thick, yafcolor3!50!white] (k1) edge (u2);
\draw[ultra thick, yafcolor3!50!white] (u2) edge (u3);

\node[circle, inner sep=2pt] (l1) at (1.5, 0.3) {};
\draw[thick,fill=yafcolor2!50!white, draw=yafcolor2] (l1) circle (3pt);
\node[anchor=west, inner sep=6pt] at (l1) {$= \ell_1$};
\node[circle, inner sep=2pt] (l1) at (1.5, -0.1) {};
\draw[thick,fill=yafcolor1!50!white, draw=yafcolor1] (l1) circle (3pt);
\node[anchor=west, inner sep=6pt] at (l1) {$= \ell_2$};
\node[circle, inner sep=2pt] (l1) at (1.5, -0.5) {};
\draw[thick,fill=yafcolor3!50!white, draw=yafcolor3] (l1) circle (3pt);
\node[anchor=west, inner sep=6pt] at (l1) {$= \ell_3$};

\end{tikzpicture}
\hspace*{\fill}

\caption{Example graphs with labels on the edges. Edge labels are indicated by colors; dashed edges indicate edges with 2 labels. Left figure: Label $\ell_1$ induces a subgraph with 6 nodes and 9 edges, and $\ell_2$ induces a subgraph with 7 nodes and 11 edges, while the conjunction of $\ell_1$ and $\ell_2$ induces a subgraph with 5 nodes and 8 edges resulting in the highest density of $8/5=1.6$. Right figure: Labels $\ell_1$, $\ell_2$, and $\ell_3$ each induce subgraphs with 5 nodes and 4 edges, while the disjunction of $\ell_1$ and $\ell_2$ induces a subgraph with 5 nodes and 7 edges, resulting in a density of $7/5=1.4$.}
\label{fig:example}
\end{figure}
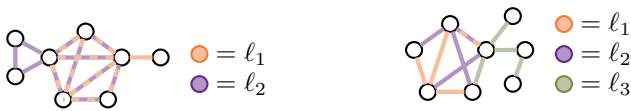

Finding the densest subgraph---with no label constraints---can be done in polynomial time~\citep{Goldberg:1984up} and can be 2-approximated in linear time~\citep{Charikar:2000tg}. Unfortunately, additional requirements on the labels
will make solving the optimization problem exactly computationally intractable:
we show that both problems are \np-hard, which forces us to resort to heuristics. We propose a greedy algorithm for both problems: we start with an empty label set and keep adding the best possible label until no additions are possible. We then return the best observed induced subgraph.

The computational bottleneck of the greedy method is selecting a new label. If done naively,
evaluating a single label candidate requires enumerating over all the edges. Since this needs to be done for every candidate during every addition, the running time is $\bigO{p \abs{L}}$, where $\abs{L}$ is the number of labels and $p$ is the number of edge-label pairs in the network. 
By keeping certain counters we can speed up the running time. We show that conjunctive-induced
graphs can be discovered in $\bigO{p \log \abs{L}}$ time using a balanced search tree,
and that disjunctive-induced
graphs can be discovered in $\bigO{p \log^2 \abs{L}}$ time with the aid of an algorithm originally used to maintain convex hulls.

This is an extended version of our previously published conference paper~\citep{kumpulainen2022community}. We extend our earlier work by considering an alternative definition of density: namely, we search for label-induced subgraphs $(W, F)$ with high $\alpha$-density $\abs{F} - \alpha \abs{W}$.
This density is closely related to the problem of finding a subgraph with maximum density~\citep{Goldberg:1984up} but also has been used to decompose 
graphs~\citep{tatti2019density,danisch2017large}. 
We show that there are $\alpha$ such that the $\alpha$-densest label-induced graph also has the highest density. We then modify the greedy algorithms to find subgraphs with high $\alpha$-density in $\bigO{p \log \abs{L}}$ for both the conjunctive and disjunctive cases.

The remainder of the paper is organized as follows.
In Section~\ref{sec:prel} we introduce the notation and formalize the optimization problem. In Sections~\ref{sec:con}--\ref{sec:dis} we present our algorithms. 
In Section~\ref{sec:alpha}, we analyze the case of using an alternative density metric and adapt the previous algorithms to this problem.
Section~\ref{sec:related} is devoted to the related work. Finally, we present the experimental evaluation in Section~\ref{sec:exp} and conclude with a discussion in Section~\ref{sec:conclusions}. The computational complexity proofs are given in Appendix~\ref{sec:proofs}.

\section{Preliminary notation and problem definition}\label{sec:prel}

In this section, we first describe the common notation and then introduce
the formal definition of our problem.

Assume that we are given an \emph{edge-labeled graph}, that is, a tuple $G=(V,E,\lab)$, where $V$ is the set of vertices, $E \subseteq \{(x,y) \mid (x, y) \in V^2, x \neq y\}$ is the set of undirected edges, and $\lab : E \rightarrow 2^L$ is a function that maps each edge $e \in E$ to the set of labels $\lab(e)$. Here $L$ is the set of all possible labels.

Given a label $\ell \in L$, let us write $E(\ell)$ to
be the edges having the label $\ell$. In addition, let us write $V(\ell)$ to be the nodes adjacent to $E(\ell)$.

Our goal is to search for dense regions of graphs that can be explained using the labels. In other words, we are looking for a set of labels that induce a dense graph. More formally, we define an \emph{inducing} function to be a function $f$ that maps two sets of labels to a binary number. An example of such a function could be $f(A; B) = [B \subseteq A]$ which returns 1 if and only if $B$ is a subset of $A$.

Given a set of labels $B \subseteq L$, an inducing function $f$, and a graph $G$, we define the \emph{label-induced subgraph} $H=G(f, B)$ as $(V(B), E(B), \lab)$, where 
\[
    E(B) = \{e \in E \mid f(\lab(e); B) = 1\}
\]
is the subset of edges that satisfy $f$, and $V(B)$ is the set of vertices that are adjacent to $E(B)$.

Given a graph $G$ with vertices $V$ and edges $E$, we measure the \emph{density} of the graph $d(G)$ as the number of edges divided by the number of vertices: $d(G) = \frac{|E|}{|V|}$.

We should point out that there are alternative choices for a notion of density. For example, one option is to consider a fraction of edges $|E|/{|V| \choose 2}$. However, this measure is however problematic since a single edge will yield a maximum value. Consequently, either a size needs to be incorporated into the objective, which leads to discovering maximum cliques---an \np-hard problem with bad approximation guarantees~\citep{DBLP:conf/focs/Hastad96}, or enumerating all pseudo-cliques with an exponential-time algorithm~\citep{Uno:2010:EAS:1712671.1712672,abello02clique}. On the other hand, finding graph with maximum $d(G)$ can be done in polynomial time~\citep{Goldberg:1984up}, and 2-approximated in linear time~\citep{Charikar:2000tg}. See related work for additional discussion.

We are now ready to state our generic problem.

\begin{problem}[\prbld]
Let $G=(V,E,\lab)$ be an edge-labeled graph over a set of labels $L$ with multiple labels being possible for each edge. Assume an inducing function $f$. Find a set of labels $L^*$ such that the density $d(H)$ of the label-induced subgraph $H=G(f, L^*)$ is maximized.
\end{problem}

We consider two special cases of \prbld. Firstly, let us define $\indand{A; B} = [B \subseteq A]$, that is, the induced edges need to contain every label in $B$. We will denote the problem \prbld paired with $\indand{}$ as \prbldand. Secondly, we define $\indor{A; B} = [B \cap A \neq \emptyset]$, that is, the induced edges need to have one common label with $B$. Then, we denote the corresponding problem as \prbldor. In other words, \prbldand is the problem of finding dense conjunctive-induced subgraphs, and \prbldor is the problem of finding disjunctive-induced subgraphs.

In addition, we consider an alternative measure to the density $d(G)$ of a graph by instead measuring the \emph{$\alpha$-density} of the graph $g(G; \alpha)$ as the number of edges minus $\alpha$ times the number of vertices: $g(G; \alpha) = |E| - \alpha|V|$. This measure is closely related to finding the densest subgraph: \citet{Goldberg:1984up} finds a series of $\alpha$-densest subgraphs when searching for the densest subgraph. However, this measure has been also studied on its own as it can be used to decompose graphs~\citep{tatti2019density,danisch2017large}.
Our optimization problem is as follows.

\begin{problem}[\prbldalpha]
Let $G=(V,E,\lab)$ be an edge-labeled graph over a set of labels $L$ with multiple labels being possible for each edge. Assume an inducing function $f$ and a constant $\alpha \in \mathbb{R}$. Find a set of labels $L^*$ such that the $\alpha$-density $g(H; \alpha)$ of the label-induced subgraph $H=G(f, L^*)$ is maximized.
\end{problem}
\section{Finding dense conjunctive-induced graphs}\label{sec:con}

In this section, we focus on \prbldand, that is, finding conjunctive-induced graphs that are dense. We will first prove that \prbldand is \np-hard.

\begin{theorem}
\prbldand is \np-hard.
\label{thr:npand}
\end{theorem}

\begin{proof}
The proof is given in Appendix~\ref{sec:proofs}.
\end{proof}

The \np-hardness forces us to resort to heuristics. Here, we use the algorithm
for 2-approximating dense subgraphs~\citep{Charikar:2000tg} as a starting point. The algorithm iteratively removes a node with the smallest degree, and returns the best solution among the observed subgraphs. We propose a similar greedy algorithm, where we greedily
add the best possible label, and repeat until the induced subgraph is empty.
We then select the best observed labels as the output.

To avoid enumerating over the edges every time we look for a new label, we maintain several counters. Let $A$ be the current set of labels. For each label, we maintain the number of nodes $n_k$ and edges $m_k$ of the candidate graph, that is,
$n_k = \abs{V(A \cup \set{k})}$ and
$m_k = \abs{E(A \cup \set{k})}$. We store the densities $m_k / n_k$ in a balanced search tree (for example, a red-black tree), which allows us to obtain the largest
element quickly. Once we update set $A$, we also update the counters and update the search tree. Maintaining the node counts $n_k$ requires us to maintain the 
counters $r_{v, k}$,  number of edges labeled as $k$ adjacent to $v$: once the counter reduces to 0, we reduce $n_k$ by 1. The pseudo-code of the algorithm is given in Algorithm~\ref{alg:greedy_conjunctive}.

\begin{algorithm}[t!]
\label{alg:greedy_conjunctive}
\caption{\alggreedyand, greedy search for the conjunctive-induced dense subgraphs}
$n_\ell \gets \abs{V(\ell)}$, for each label $\ell \in L$\;
$m_\ell \gets \abs{E(\ell)}$, for each label $\ell \in L$\;
$r_{v,\ell} \gets \abs{\set{e \in E(\ell) \mid e \text{ is adjacent to } v}}$, for each vertex $v$ and label $\ell$\;
$T \gets$ labels sorted by the density values $\frac{m_\ell}{n_\ell}$ (e.g., in a red-black tree)\; \label{lst:line:condensity}
$A_0 \gets \emptyset$ and $i \gets 0$\;
\While{there are labels} {
pick and remove label $k$ that has the maximum density in $T$\;
$A_{i+1} \gets A_i \cup \{k\}$\;
\For{each edge $e$ without label $k$} { \nllabel{conjunctive_for1}
    \For{each label $\ell$ of edge $e = (u, v)$} { \nllabel{conjunctive_for2}
        $m_\ell \gets m_\ell - 1$\;
        $r_{v,\ell} \gets r_{v,\ell}-1$;
        $r_{u,\ell} \gets r_{u,\ell}-1$\;
        \lIf {$r_{v,\ell} = 0$} {$n_\ell \gets n_\ell-1$} 
        \lIf {$r_{u,\ell} = 0$} {$n_\ell \gets n_\ell-1$} 
    }
    remove edge $e$\;
}
update $T$ for all labels $\ell$ with changed values of $m_\ell$ or $n_\ell$\;
$i \gets i+1$\;
}
\Return the set of labels $A_i$ that yields the highest density\;
\end{algorithm}

We conclude with an analysis of the computational complexity of \alggreedyand.

\begin{theorem}
\label{thr:contime}
\alggreedyand runs in $\bigO{p\log{\abs{L}}+\abs{V}+\abs{E}}$ time, where
$p$ is the number of edge-label pairs $p=\abs{\{(e,k)\mid e\in E, k\in \lab(e)\}}$.
\end{theorem}

\begin{proof}
Initializing counters in \alggreedyand can be done
in $\bigO{\abs{V} + \abs{E} + \abs{L}}$ time while
initializing the tree can be done in $\bigO{\abs{L} \log \abs{L}}$ time.

Let us consider the inner for-loop. Since an edge is deleted once it is processed, the inner for-loop is executed at most
$p$ times during the search. Since this is the only way the counters get updated, the tree $T$ is updated $p$ times,
each update requiring $\bigO{\log \abs{L}}$ time.

The outer loop is executed at most $\abs{L}$ times. During each round, selecting and removing the label requires $\bigO{\log \abs{L}}$ time.

In summary, the algorithm requires 
\[
\bigO{\abs{V} + \abs{E} + \abs{L} + \abs{L} \log \abs{L} + p \log \abs{L}} \subseteq
\bigO{\abs{V} + \abs{E} + p \log \abs{L}}
\]
time, completing the proof.
\end{proof}
\section{Finding dense disjunctive-induced graphs}\label{sec:dis}

In this section, we focus on \prbldor, that is, finding disjunctive-induced graphs that are dense. We will first prove that \prbldor is \np-hard.

\begin{theorem}
\prbldor is \np-hard.
\label{thr:npor}
\end{theorem}

\begin{proof}
The proof is given in Appendix~\ref{sec:proofs}.
\end{proof}

Similar to \prbldand, we resort to a greedy search to find good subgraphs: We start with an empty label set, and iteratively add the best possible label. Once done, we return the best observed label set.

However, we maintain a different set of counters as compared to \alggreedyand.
The reason for having different counters is to avoid a significantly higher number of updates: the inner loop would need to go over the edge-label pairs that are \emph{not} present in the graph. More formally, we maintain values $n$ and $m$ representing the number of nodes and edges in the subgraph induced by the current set of labels, say $A$. We also maintain $n_k$ and $m_k$,
the number of \emph{additional} nodes and edges if $k$ is added to $A$. At each iteration, we select the label optimizing $\frac{m + m_k}{n + n_k}$. We will discuss the selection process later. Once the label is selected, we update the counters $m_k$ and $n_k$. To maintain $n_k$ properly, we keep track of what nodes are already in $V(A)$, using an indicator $r_v$ with $r_v = 1$ if $v \in V(A)$. The pseudo-code for the algorithm is given in Algorithm~\ref{alg:greedy_disjunctive}.

\begin{algorithm}[t!]
\label{alg:greedy_disjunctive}
\caption{\alggreedyor, greedy search for the disjunctive-induced dense subgraphs}
$n \gets 0$;
$m \gets 0$\;
$n_\ell \gets \abs{V(\ell)}$, for each label $\ell \in L$\;
$m_\ell \gets \abs{E(\ell)}$, for each label $\ell \in L$\;
$S_{v} \gets \set{\ell \in L \mid \text{there is an edge with label $\ell$ adjacent to $v$}}$\;
$r_{v} \gets 0$, for each vertex $v$\;
$A_0 \gets \emptyset$ and $i \gets 0$\;
\While{there are labels} {
pick and remove label $k$ that yields the maximum density $\frac{m+m_k}{n+n_k}$\;
$A_{i+1} \gets A_i \cup \{k\}$\;
\For{each edge $e = (u, v)$ with label $k$} {
    \lFor{each label $\ell$ of edge $e = (u, v)$} {
        $m_\ell \gets m_\ell - 1$
    }
    $m \gets m + 1$\;
    \If {$r_{v} = 0$} {
        \lFor{each label $\ell$ in $S_{v}$} {
            $n_\ell \gets n_\ell-1$
        }
        $n \gets n+1$\;
    } 
    \If {$r_{u} = 0$} {
        \lFor{each label $\ell$ in $S_{u}$} {
            $n_\ell \gets n_\ell-1$
        }
        $n \gets n+1$\;
    } 
    $r_{v} \gets 1$;
    $r_{u} \gets 1$\;
    remove edge $e$\;
}
$i \gets i+1$\;
}
\Return the set of labels $A_i$ that yields the highest density\;
\end{algorithm}

During each iteration, we need to select the label maximizing $\frac{m+m_k}{n+n_k}$.
We cannot use priority queues any longer since $n$ and $m$ change every iteration.
However, we can speed up the selection using a convex hull, a classic concept from computational geometry, see for example,~\citep{Li2011}. First, let us formally define a lower-right convex hull.

\begin{definition}
Given a set of points $X = \set{(x_i, y_i)}$ in a plane,
we define a \emph{lower-right convex hull} $H  = \hull{H}$ to
be a subset of $X$ such that $q = (x_q, y_q) \in X$ is \emph{not}
in $X$ if and only if there is a point $r = (x_r, y_r) \in H$ such that
$x_q \leq x_r$ and $y_q \geq y_r$, or if there are two points $p, r \in H$
such that $q$ is above or at the segment joining $q$ and $r$.
\end{definition}

If we were to plot $X$ on a plane, then $\hull{X}$ is the lower-right portion of
the complete convex hull, that is, a set of points in $X$ that form a convex polygon
containing $X$. For notational simplicity, we will refer to $\hull{X}$
as the convex hull. Note that if we order the points in $\hull{X}$ by their $x$-coordinates, then
the $y$-coordinates and the slopes of the intermediate segments are also increasing.

We will first argue that we only need to search the convex hull when looking for the optimal label.

\begin{theorem}
Let $X$ be a set of positive points ${(m_i, n_i)}$, and let $H = \hull{X}$ be the convex hull. Select $m, n \geq 0$.
Then $\max_{p \in X} \frac{m+m_i}{n+n_i} = \max_{p \in H} \frac{m+m_i}{n+n_i}$.
\label{thr:optimal_on_convex_hull}
\end{theorem}

\begin{proof}
Let $k = (m_k, n_k)$ be the optimal point in $X$.
Assume that $k \notin H$.
Assume that there is a point $q = (m_q, n_q)$ in $H$
such that $m_q \geq m_k$ and $n_q \leq n_k$. Then $\frac{m+m_k}{n+n_k} \leq \frac{m+m_q}{n+n_q}$, so the point $q$ is also optimal. 

Assume there is no such point $q$.
Then, the $x$-coordinate of point $k$ falls between two consecutive points $p$ and $q$ in $H$, that is, $m_p < m_k < m_q$. Then $k$ must be above the segment between $p$ and $q$ as otherwise, $k$ would also be a part $H$. Therefore, the slope for the segment between $p$ and $k$ must be greater than the slope of the segment between $p$ and $q$, and the slope for the segment between $k$ and $q$ must be smaller,  
\begin{equation}
\frac{n_q-n_k}{m_q-m_k} \leq \frac{n_q-n_p}{m_q-m_p} \leq \frac{n_k-n_p}{m_k-m_p}\label{slopes}.
\end{equation}

Furthermore, since $k \notin H$, we must have $n_k > n_p$.
By assumption, we also have $n_k < n_q$. In summary, we have $n_p < n_k < n_q$ and $m_p < m_k < m_q$, which means that the slopes in Equation~\ref{slopes} are all positive. By taking the reciprocals this then gives,

\begin{equation}
\frac{m_q-m_k}{n_q-n_k} \geq \frac{m_q-m_p}{n_q-n_p} \geq \frac{m_k-m_p}{n_k-n_p}.
\label{inverse_slopes}
\end{equation}

Denote then the objective value at point $k$ by $c = \frac{m+m_k}{n+n_k}$.
Let $x_1 = c(n+n_p) - m$.
Then, the optimality of $k$ implies $\frac{m+x_1}{n+n_p} = c \geq \frac{m+m_p}{n+n_p}$, which means $x_1 \geq m_p$.
The definition of $c$ leads to $m = c(n+n_k) - m_k$, which in turns leads
to $x_1 = c(n_p-n_k) + m_k$. Solving for $c$ we get $c = \frac{m_k-x_1}{n_k-n_p}$.
Substituting $x_1 \geq m_p$ yields $c \leq \frac{m_k-m_p}{n_k-n_p}$, using Equation~\ref{inverse_slopes} then yields $c \leq \frac{m_q-m_k}{n_q-n_k}$.

Next, let $x_2 = c(n_q-n_k) + m_k$ which means that $c = \frac{x_2-m_k}{n_q-n_k}$.
Now since $c \leq \frac{m_q-m_k}{n_q-n_k}$ we must have $x_2 \leq m_q$.
Since $m_k = c(n+n_k) - m$, we also have $x_2 = c(n_q+n) - m$, yielding
$c = \frac{m+x_2}{n+n_q} \leq \frac{m+m_q}{n+n_q}$, thus $q$ is also optimal.
\end{proof}

Theorem~\ref{thr:optimal_on_convex_hull} states that we need to only consider the convex hull $H$ of the set $\set{(m_i, n_i)}$ when searching for the optimal new label. Note that $H$ does not depend on $n$ or $m$. Moreover, we can use the algorithm by~\citet{overmars1981maintenance} to maintain $H$ as $n_k$ and $m_k$ are updated in $\bigO{\log^2 \abs{L}}$ time per update. We will see that the number of needed updates is bounded by the number of edge-label pairs.

However, the convex hull can be as large as the original set, so our goal is to avoid enumerating over the whole set. To this end, we design a binary search strategy over the hull. We will first introduce two quantities used in our search.

\begin{definition}
Given two points $p, q \in \hull{X}$, we define the inverse slope as $s(p,q) = \frac{m_q-m_p}{n_q-n_p}$ and the bias term as $b(p,q) = \frac{m_q n_p - m_p n_q}{n_q-n_p}$.
\end{definition}

First, let us prove that both $s$ and $b$ are monotonically decreasing. 

\begin{lemma}
\label{bound_decreasing}
Let $p$, $q$, and $r$ be three consecutive points in $\hull{X}$. Then we have $n \times s(q,r) + b(q,r) \leq n \times s(p,q) + b(p,q)$, for any $n \geq 0$.
\end{lemma}

\begin{proof}
The slope for the segment between $p$ and $q$ is less than or equal to the slope for the segment between $q$ and $r$. Inversing the slopes leads to
\[
    s(q, r) = 
    \frac{m_r-m_q}{n_r-n_q}
    \leq
    \frac{m_q-m_p}{n_q-n_p}
    = s(p, q).
\] 
By cross-multiplying, adding $m_q n_q - m_q n_p - m_q n_r + \frac{m_q n_p n_r}{n_q}$ to both sides,  multiplying by $\frac{n_q}{(n_r-n_q)(n_q-n_p)}$, and simplifying, we get  
\[
    b(q,r) = \frac{m_r n_q - m_q n_r}{n_r-n_q} \leq \frac{m_q n_p - m_p n_q}{n_q-n_p} = b(p,q).
\]

Combining the two equations proves the claim.
\end{proof}

Next, we show the key necessary condition for the optimal point.

\begin{lemma}
Let $p$, $q$, and $r$ be 3 consecutive points in $\hull{X}$.
Select $n, m \geq 0$.
If $q$ optimizes $\frac{m_q + m}{n_q + n}$, then $n \times s(q,r) + b(q,r) \leq m \leq n \times s(p,q) + b(p,q)$.
\end{lemma}

\begin{proof}
Since $q$ is optimal, we have $\frac{m+m_p}{n+n_p} \leq \frac{m+m_q}{n+n_q}$. Solving for $m$ gives us
$m \leq n\frac{m_q-m_p}{n_q-n_p} + \frac{m_q n_p - m_p n_q}{n_q-n_p} = n \times s(p,q) + b(p,q)$. Similarly, due to optimality, $\frac{m+m_r}{n+n_r} \leq \frac{m+m_q}{n+n_q}$, and solving for $m$ leads to
$m \geq  n \times s(p,q) + b(p,q)$, proving the claim.
\end{proof}

The two lemmas allow us to use binary search as follows. Given two consecutive points $p$ and $q$ we test whether $m \leq n \times s(p,q) + b(p,q)$. If true, then the optimal label is $q$ or to the right of $q$, if false, the optimal point is to the left of $q$. To perform the binary search, we can use directly the structure maintained by the algorithm by~\citet{overmars1981maintenance} since it stores the current convex hull in a balanced search tree. Moreover, the algorithm allows evaluating any function based on the neighboring points. Specifically, we can maintain $s$ and $b$. In summary, we can find the optimal label in $\bigO{\log \abs{L}}$ time.

Our next result formalizes the above discussion.

\begin{theorem}
\alggreedyor runs in $\bigO{p\log^2{\abs{L}}+\abs{V}+\abs{E}}$ time, where
$p$ is the number of edge-label pairs $p=\abs{\{(e,k)\mid e\in E, k\in \lab(e)\}}$.
\end{theorem}

\begin{proof}
The proof is similar to the proof of Theorem~\ref{thr:contime}, except we have
replaced a search tree with the convex hull structure by~\citet{overmars1981maintenance}.
The inner for-loops are evaluated at most $\bigO{p}$ times since an edge or a node is visited only once, and $\sum_v \abs{S_v} \in \bigO{p}$. Maintaining the hull requires $\bigO{\log^2 \abs{L}}$ time, and there are at most $\bigO{p}$ such updates. Searching for an optimal label requires $\bigO{\log \abs{L}}$ time, and there are at most $\abs{L}$ such searches.
\end{proof}

We should point out that a faster algorithm by~\citet{brodal2002dynamic} maintains the convex hull in $\bigO{\log \abs{L}}$ time. However, this algorithm does not provide a search tree structure that we can use to search for the optimal addition.
\section{Finding subgraphs with high $\alpha$-density}\label{sec:alpha}

In this section, we focus on the problem \prbldalpha of finding subgraphs with high $\alpha$-density. 

The following classic result in fractional programming~\citep{dinkelbach1967nonlinear} shows how the problem of finding the maximum density subgraph reduces to maximizing the $\alpha$-density of a subgraph for a large enough value of $\alpha$. An immediate consequence of this result is that solving \prbldalpha is \np-hard.

\begin{theorem}
\label{thr:alphamax}
Write $H_\alpha$ to be the solution to \prbldalpha.
There is $\tau$ such that $H_\tau$ also solves \prbld.
Moreover, for any $\alpha > \tau$, the graph $H_\alpha$
either solves \prbld or is empty.
\end{theorem}

\begin{proof}
Let $H^*$ be a solution to \prbld with $\sigma = d(H^*)$.
Since there are a finite number of subgraphs, there is $\tau < \sigma$
such that any graph $H$ with $d(H) \geq \tau$ has $d(H) = \sigma$.

Since $g(H^*; \tau) > 0$, we have $g(H_\tau; \tau) > 0$, or $\abs{E(H_\tau)} - \tau \abs{V(H_\tau)} > 0$ which implies $d(H_\tau) > \tau$. By definition of $\tau$, the subgraph $H_\tau$ solves \prbld.

Similarly, for any $\alpha > \tau$, we have $g(H_\alpha; \alpha) \geq 0$. Consequently, either $H_\alpha$ is empty or $d(H_\alpha) \geq \alpha > \tau$, that is, $H_\alpha$ solves \prbld.
\end{proof}

\begin{corollary}
\label{cor:npalpha}
\prbldalpha is \np-hard for both $\indor{}$ and $\indand{}$.
Moreover, both problems are inapproximable unless $\poly=\np$.
\end{corollary}

\begin{proof}
The proof is given in Appendix~\ref{sec:proofs}.
\end{proof}

To find solutions to \prbldalpha in practice, we adapt the previous greedy algorithms to find subgraphs with high $\alpha$-density. In the conjunctive case, we get the $\alggreedyandalpha$ algorithm by simply changing the density on line~\ref{lst:line:condensity} of Algorithm~\ref{alg:greedy_conjunctive} from $\frac{m_\ell}{n_\ell}$ to $m_\ell-\alpha n_\ell$. This leads to the same computational complexity as for \alggreedyand.

In the disjunctive case, we again keep track of the counters to find the number of additional nodes and edges when a label is added to the current set of labels. However, the $\alpha$-density to maximize now becomes $(m+m_k)-\alpha(n+n_k)$. As $m-\alpha n$ does not depend on the label, we only need to find the label $k$ that maximizes $m_k-\alpha n_k$. We may thus use a balanced search tree as in the conjunctive case. The pseudo-code for this algorithm is given in Algorithm~\ref{alg:greedy_disjunctive_alpha}. 

\begin{algorithm}[t!]
\label{alg:greedy_disjunctive_alpha}
\caption{\alggreedyoralpha, greedy search for the disjunctive-induced subgraphs with high $\alpha$-density}
$n \gets 0$;
$m \gets 0$\;
$n_\ell \gets \abs{V(\ell)}$, for each label $\ell \in L$\;
$m_\ell \gets \abs{E(\ell)}$, for each label $\ell \in L$\;
$S_{v} \gets \set{\ell \in L \mid \text{there is an edge with label $\ell$ adjacent to $v$}}$\;
$r_{v} \gets 0$, for each vertex $v$\;
$T \gets$ labels sorted by the density values $m_\ell-\alpha n_\ell$ (e.g., in a red-black tree)\;
$A_0 \gets \emptyset$ and $i \gets 0$\;
\While{there are labels} {
pick and remove label $k$ that has the maximum density in $T$\;
$A_{i+1} \gets A_i \cup \{k\}$\;
\For{each edge $e = (u, v)$ with label $k$} {
    \lFor{each label $\ell$ of edge $e = (u, v)$} {
        $m_\ell \gets m_\ell - 1$
    }
    $m \gets m + 1$\;
    \If {$r_{v} = 0$} {
        \lFor{each label $\ell$ in $S_{v}$} {
            $n_\ell \gets n_\ell-1$
        }
        $n \gets n+1$\;
    } 
    \If {$r_{u} = 0$} {
        \lFor{each label $\ell$ in $S_{u}$} {
            $n_\ell \gets n_\ell-1$
        }
        $n \gets n+1$\;
    } 
    $r_{v} \gets 1$;
    $r_{u} \gets 1$\;
    remove edge $e$\;
}
update $T$ for all labels $\ell$ with changed values of $m_\ell$ or $n_\ell$\;
$i \gets i+1$\;
}
\Return the set of labels $A_i$ that yields the highest density\;
\end{algorithm}

As \alggreedyoralpha does not need to use a convex hull but uses a balanced search tree instead, the running time becomes the same as for the conjunctive case.

\begin{theorem}
\label{thr:alphatime}
\alggreedyandalpha and \alggreedyoralpha run in $\bigO{p\log{\abs{L}}+\abs{V}+\abs{E}}$ time, where
$p$ is the number of edge-label pairs $p=\abs{\{(e,k)\mid e\in E, k\in \lab(e)\}}$.
\end{theorem}

\begin{proof}
The proofs for both cases are virtually the same as the proof of Theorem~\ref{thr:contime}.
\end{proof}

We conclude this section by considering the (lack of the) hierarchy property of $\alpha$-density.
\citet{tatti2019density} showed that the subgraphs (without label constraints) optimizing $g(\cdot, \alpha)$ form a nested
structure, that is, if we write $H_\alpha$ to be the optimal solution, then $H_\beta \subseteq H_\alpha$ for any $\beta > \alpha$.
Such a decomposition may be useful as it partitions the nodes into increasingly dense regions. Unfortunately, this is not the case for us as shown in Figure~\ref{fig:counter}.

\begin{figure}
\hspace*{\fill}
\begin{tikzpicture}

\node[circle, inner sep=2pt] (v1) at (0, 0) {};
\draw[thick,fill=white, draw=black] (v1) circle (3pt);

\node[circle, inner sep=2pt] (v2) at (-0.3, -0.5) {};
\draw[thick,fill=white, draw=black] (v2) circle (3pt);

\node[circle, inner sep=2pt] (v3) at (0.3, -0.5) {};
\draw[thick,fill=white, draw=black] (v3) circle (3pt);

\begin{scope}[xshift=1cm, yshift=-0.2cm]
\node[circle, inner sep=2pt] (u1) at (0:0.5) {};
\draw[thick,fill=white, draw=black] (u1) circle (3pt);

\node[circle, inner sep=2pt] (u2) at (1*72:0.5) {};
\draw[thick,fill=white, draw=black] (u2) circle (3pt);

\node[circle, inner sep=2pt] (u3) at (2*72:0.5) {};
\draw[thick,fill=white, draw=black] (u3) circle (3pt);

\node[circle, inner sep=2pt] (u4) at (3*72:0.5) {};
\draw[thick,fill=white, draw=black] (u4) circle (3pt);

\node[circle, inner sep=2pt] (u5) at (4*72:0.5) {};
\draw[thick,fill=white, draw=black] (u5) circle (3pt);
\end{scope}

\draw[ultra thick, yafcolor2!50!white] (v1) edge[bend left=10] (v2);
\draw[ultra thick, yafcolor2!50!white] (v1) edge[bend left=10] (v3);
\draw[ultra thick, yafcolor2!50!white] (v2) edge[bend left=10] (v3);

\draw[ultra thick, yafcolor1!50!white] (u1) edge[bend left=10] (u2);
\draw[ultra thick, yafcolor1!50!white] (u2) edge[bend left=10] (u3);
\draw[ultra thick, yafcolor1!50!white] (u3) edge[bend left=10] (u4);
\draw[ultra thick, yafcolor1!50!white] (u4) edge[bend left=10] (u5);

\node[circle, inner sep=2pt] (l1) at (2, 0) {};
\draw[thick,fill=yafcolor2!50!white, draw=yafcolor2] (l1) circle (3pt);
\node[anchor=west, inner sep=6pt] at (l1) {$= \ell_1$};
\node[circle, inner sep=2pt] (l1) at (2, -0.4) {};
\draw[thick,fill=yafcolor1!50!white, draw=yafcolor1] (l1) circle (3pt);
\node[anchor=west, inner sep=6pt] at (l1) {$= \ell_2$};

\end{tikzpicture}\hfill
\begin{tikzpicture}
\begin{scope}[yshift=-0.2cm]
\node[circle, inner sep=2pt] (k1) at (0*60:0.5) {};
\draw[thick,fill=white, draw=black] (k1) circle (3pt);

\node[circle, inner sep=2pt] (k2) at (1*60:0.5) {};
\draw[thick,fill=white, draw=black] (k2) circle (3pt);

\node[circle, inner sep=2pt] (k3) at (2*60:0.5) {};
\draw[thick,fill=white, draw=black] (k3) circle (3pt);

\node[circle, inner sep=2pt] (k4) at (3*60:0.5) {};
\draw[thick,fill=white, draw=black] (k4) circle (3pt);

\node[circle, inner sep=2pt] (k5) at (4*60:0.5) {};
\draw[thick,fill=white, draw=black] (k5) circle (3pt);

\node[circle, inner sep=2pt] (k6) at (5*60:0.5) {};
\draw[thick,fill=white, draw=black] (k6) circle (3pt);
\end{scope}

\node[circle, inner sep=2pt, right=0.3cm of k1] (u1) {};
\draw[thick,fill=white, draw=black] (u1) circle (3pt);

\node[circle, inner sep=2pt, above left=0.1 and 0.3 of k4] (v1) {};
\draw[thick,fill=white, draw=black] (v1) circle (3pt);

\node[circle, inner sep=2pt, below left=0.1 and 0.3 of k4] (v2) {};
\draw[thick,fill=white, draw=black] (v2) circle (3pt);

\foreach \x [count=\xi from 2] in {1,...,5}
\foreach \y in {\xi,...,6} {
\draw[ultra thick, yafcolor1!50!white] (k\x) edge[bend left=10] (k\y);
\draw[ultra thick, dashed, yafcolor2!50!white] (k\x) edge[bend left=10] (k\y);
}

\draw[ultra thick, yafcolor2!50!white] (k1) edge[bend left=10] (u1);
\draw[ultra thick, yafcolor1!50!white] (k4) edge[bend left=10] (v1);
\draw[ultra thick, yafcolor1!50!white] (v1) edge[bend left=10] (v2);
\draw[ultra thick, yafcolor1!50!white] (v2) edge[bend left=10] (k4);

\node[circle, inner sep=2pt] (l1) at (1.5, 0) {};
\draw[thick,fill=yafcolor2!50!white, draw=yafcolor2] (l1) circle (3pt);
\node[anchor=west, inner sep=6pt] at (l1) {$= \ell_1$};
\node[circle, inner sep=2pt] (l1) at (1.5, -0.4) {};
\draw[thick,fill=yafcolor1!50!white, draw=yafcolor1] (l1) circle (3pt);
\node[anchor=west, inner sep=6pt] at (l1) {$= \ell_2$};

\end{tikzpicture}
\hspace*{\fill}

\caption{Subgraphs with optimal $\alpha$-density are not nested.
Left figure: $\ell_1$ is optimal for $\alpha = 3/4$
and $\ell_2$ is optimal for $\alpha = 1/4$ when using $\indand{}$.
Right figure: $\ell_1$ is optimal for $\alpha = 2.25$
and $\ell_2$ is optimal for $\alpha = 1.75$ when using $\indor{}$.}
\label{fig:counter}
\end{figure}
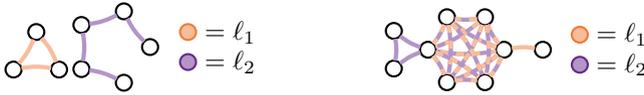

Interestingly enough, if we allow more flexible queries, we can show that we too obtain a nested structure. More formally, given a Boolean formula $B$ we define $G(B)$ to be the subgraph consisting of edges whose labels satisfy $B$, and the incident vertices. Then the optimization problem would be to find the Boolean formula $B$ maximizing $g(G(B); \alpha)$. We then have the following proposition.

\begin{proposition}
Let $H_\alpha$ be a subgraph induced by a Boolean formula $B_\alpha$ that optimizes
$g(\cdot; \alpha)$.
Then $H_\alpha \subseteq H_\beta$ for any $\alpha > \beta$.
\end{proposition}

\begin{proof}
Assume otherwise. Write $X = H_\alpha \cup H_\beta$ and $Y = H_\alpha \cap H_\beta$. Note that $X$ is induced by $B_\alpha \lor B_\beta$ and $Y$ is induced by $B_\alpha \land B_\beta$. Then
\[
    g(X; \beta) - g(H_\beta; \beta) > 
    g(X; \alpha) - g(H_\beta; \alpha) =
    g(H_\alpha; \alpha) - g(Y; \alpha) \geq 0,
\]
where the last inequality is due to the optimality of $H_\alpha$.
Thus, $g(X; \beta) > g(H_\beta; \beta)$ violating
the optimality of $H_\beta$.
\end{proof}
\section{Related work}\label{sec:related}

A closely related work to our method is an approach proposed by~\citet{galbrun2014overlapping}.
Here the authors search for multiple dense subgraphs that can be explained by conjunction
on (or the majority of) the \emph{node} labels. The authors propose a greedy algorithm for finding such subgraphs. Interestingly enough, the authors do not show that the underlying problem is \np-hard---although we conjecture that this is indeed the case---instead, they show that the subproblem arising from the greedy approach is an \np-hard problem.

Another closely related work is an approach proposed by~\citet{pool2014description}, where the authors
search for dense subgraphs that can be explained by queries on the \emph{nodes}. The quality of the subgraphs is a ratio $S/C$, where $S$ measures the goodness of a subgraph using the edges within the subgraph as well as the cross-edges, and $C$ measures the complexity of the query.

The major difference between our work and the aforementioned work is that our method uses labels on the edges. While conceptually a small difference, this distinction leads to different algorithms and different analyses of those algorithms. Moreover, we cannot apply directly the previously discussed methods to networks that only have labels on edges.

An appealing property of finding subgraphs that maximize $\abs{E(W)}/\abs{W}$, or equivalently an average degree, is that we can find the optimal solution in polynomial time~\citep{Goldberg:1984up}. Furthermore, we can 2-approximate the graph with a simple linear algorithm~\citep{Charikar:2000tg}. The algorithm iteratively removes the node with the smallest degree and then selects the best available graph. This algorithm is essentially the same as the algorithm used to discover $k$-cores, subgraphs that have the \emph{minimum} degree of at least $k$. The connection between the $k$-cores and dense subgraphs is further explored by~\citet{tatti2019density}, where the dense subgraphs are extended to create an increasingly dense structure.
A variant of a quality measure was proposed by~\citet{tsourakakis15triangle}, where the quality of the subgraph is the ratio of triangles over the vertices. In another variant by~\citet{bonchi2019distance}, the edges were replaced with paths of at most length $k$. Finding such structures in labeled graphs poses an interesting line of future work.

While finding dense subgraphs is polynomial, finding cliques is an \np-hard problem with a very strong inapproximability bound~\citep{DBLP:conf/focs/Hastad96}. Finding cliques may be impractical as they do not allow any absent edges. To relax the requirement, \citet{abello02clique} and \citet{Uno:2010:EAS:1712671.1712672} proposed
searching for quasi-cliques, that is subgraphs with a high proportion of edges, $\abs{E(W)} / \binom{\abs{W}}{2}$.
Another relaxation of cliques is $k$-plex where $k$ absent edges are allowed for a vertex~\citep{seidman1983network}. Finding $k$-plexes remain an \np-hard problem~\citep{balasundaram:2011:kflex}. Alternatively, we can relax the definition by considering $n$-cliques, where vertices must be connected with an $n$-path~\citep{Bron:1973:AFC:362342.362367}, or $n$-clans where we also require that the diameter of the graph is $n$~\citep{mokken1979cliques}. Since $1$-clique (and $1$-clan) is a clique, these problems remain computationally intractable.
\section{Experimental evaluation}\label{sec:exp}

In this section, we describe our experimental evaluation of the \alggreedyand and \alggreedyor algorithms. First, we observe how the algorithms behave on synthetic data with increasing randomness. Then we apply the algorithms to real-world datasets and analyze the results. 

We implement our algorithms in Python and the source code is available online
\footnote{\url{https://version.helsinki.fi/dacs}}.
Since the number of labels in our experiments was not exceedingly large, we did not use the speed up using convex hulls when implementing disjunctive-induced graphs. Instead, we search for the optimal label from scratch leading to a running time of $\bigO{p\abs{L}}$.

\textbf{Experiments with synthetic data:}
We evaluate the greedy algorithms on synthetic graphs of 200 vertices and 50 labels. We select 5 of the labels as target labels and construct graphs for the conjunctive and disjunctive cases such that selecting the subgraph induced by these 5 labels gives the best density. We then add random noise to the network by introducing a noise parameter $\epsilon$, which controls the probability of randomly adding and removing edges as well as adding new labels to the edges.

For the conjunctive case, we create five disjoint cliques of 10 vertices such that all edges on the $k$th clique have all except the $k$th of the target labels. Finally, we add one more 20 vertex clique that has all of the target labels. Since each of the smaller cliques is missing one of the target labels, selecting the conjunction of all of them yields the densest subgraph as the clique of 20 vertices. 

Given the noise parameter $\epsilon$, we then add noise by having each of the edges in the cliques removed with probability $\epsilon$, as well as having any other edges added between any pair of vertices with probability $\epsilon$. Finally, for each of the edges in the cliques, we add any of the other labels with probability $\epsilon$ each, except for adding the remaining target labels to edges in the cliques. 

For the disjunctive case, we have created one clique with 40 vertices. The edges in the clique are split into five sets, such that each set of edges gets one of the target labels. Now, selecting the disjunction of the five target labels induces the clique as the subgraph and results in the highest density.

We then add noise by removing edges from the clique and adding new edges between any other pair of vertices with probability $\epsilon$. In addition, each edge gains any of the other labels also with probability $\epsilon$.

\begin{figure}[t!]
\centerline{\includegraphics[width=0.5\textwidth]{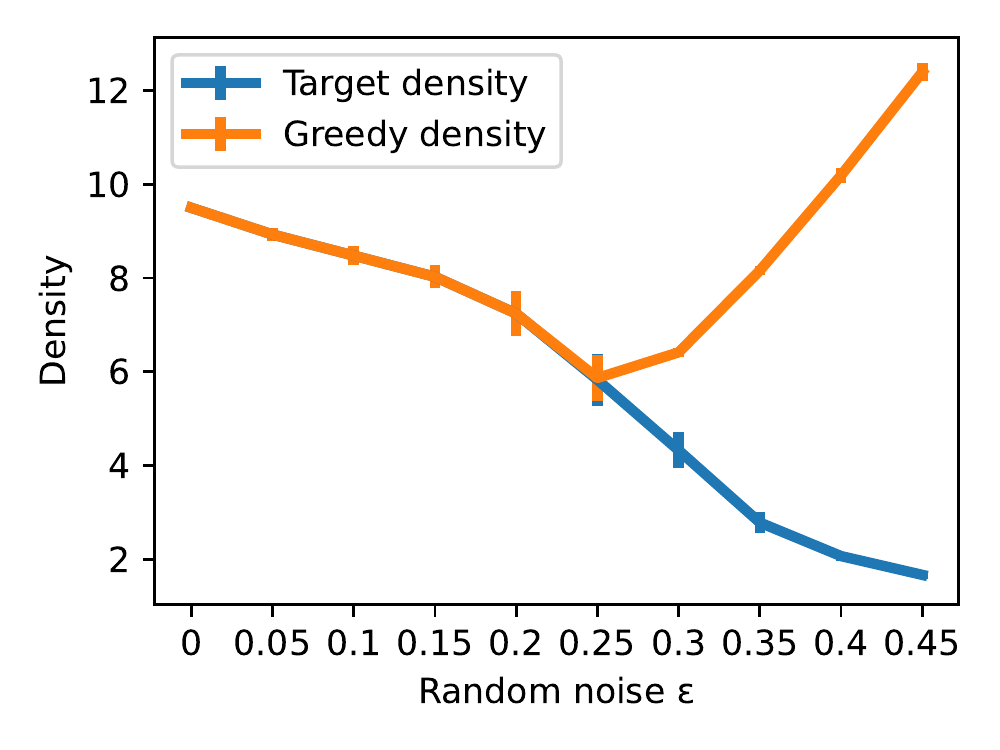}%
\includegraphics[width=0.5\textwidth]{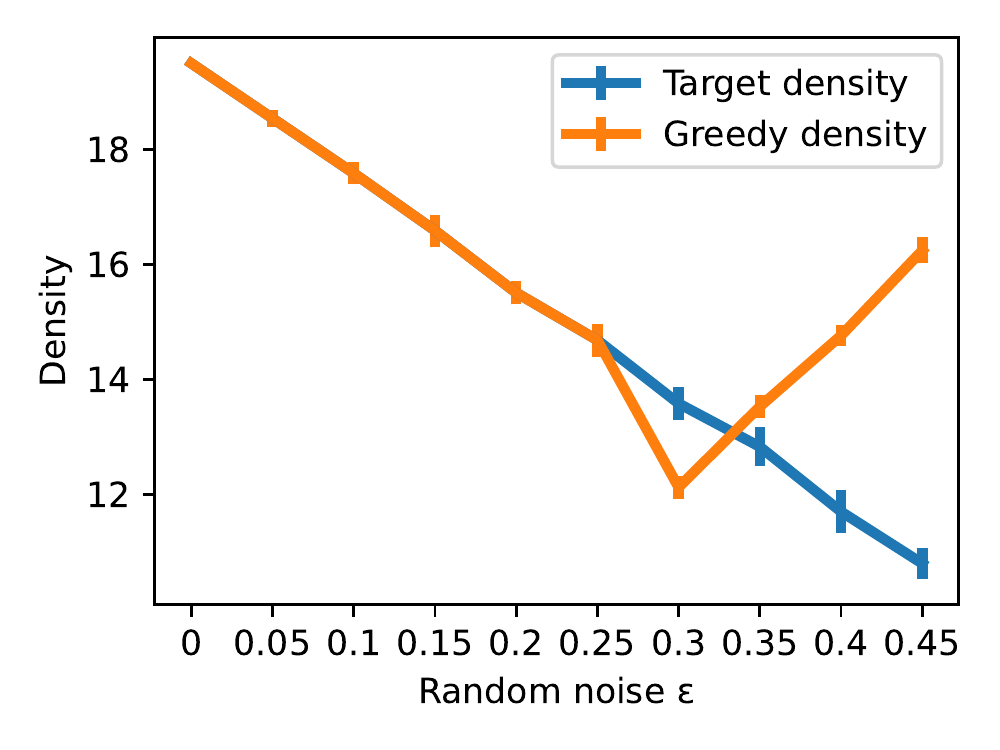}
}
\caption{Density of the subgraph induced by the target labels and the subgraph induced by the labels chosen by the greedy algorithms as a function of noise $\epsilon$ in the network. The line shows the mean density of 10 runs for each $\epsilon$, and the vertical error bars show their standard deviation. The results for \alggreedyand algorithm are on the left and for \alggreedyor on the right.}
\label{synthetic}
\end{figure}

\begin{figure}[t!]
\centerline{\includegraphics[width=0.5\textwidth]{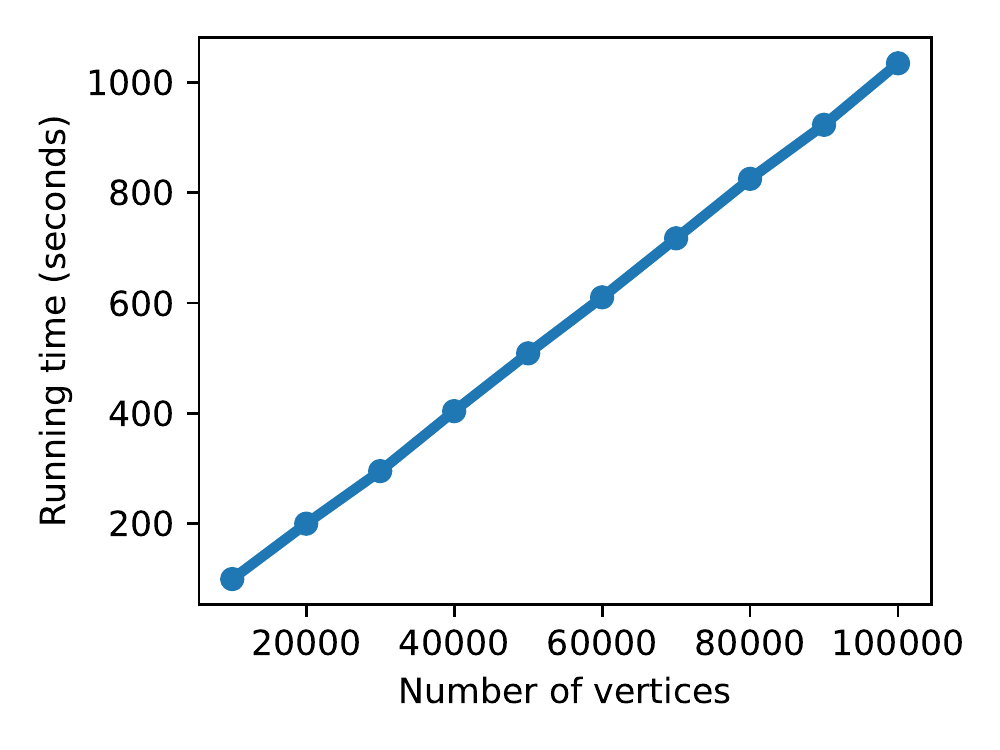}%
\includegraphics[width=0.5\textwidth]{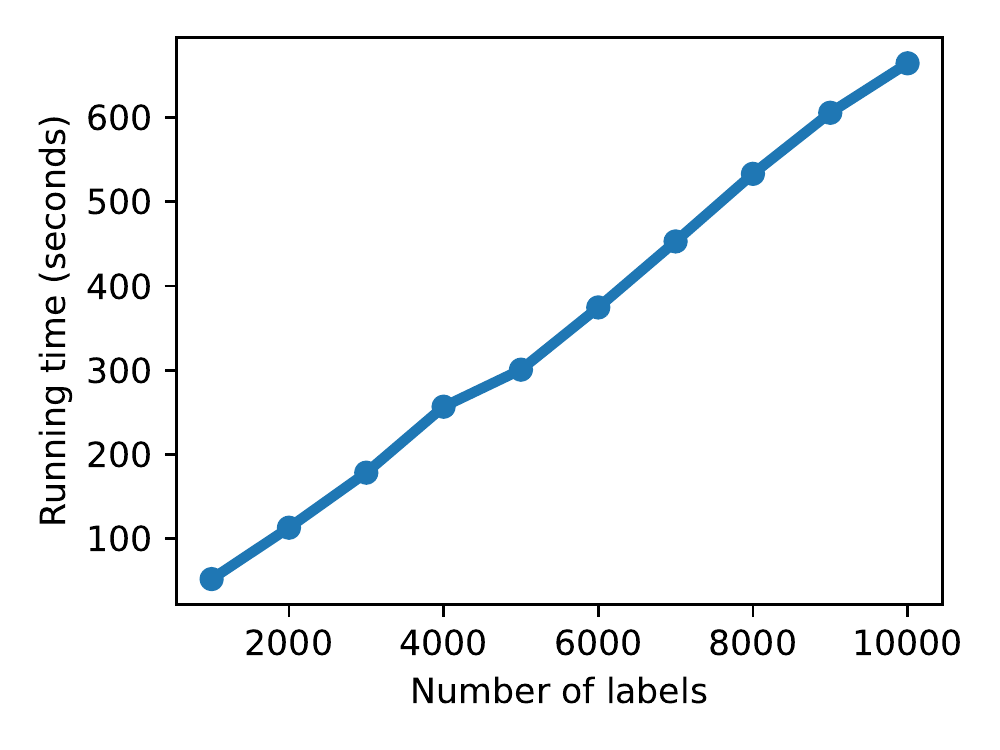}
}
\caption{Running time of the \alggreedyand algorithm as a function of the number of vertices (left) and the number of labels (right) in our synthetic graphs.}
\label{runtime_conjunction}
\end{figure}

\begin{figure}[t!]
\centerline{\includegraphics[width=0.5\textwidth]{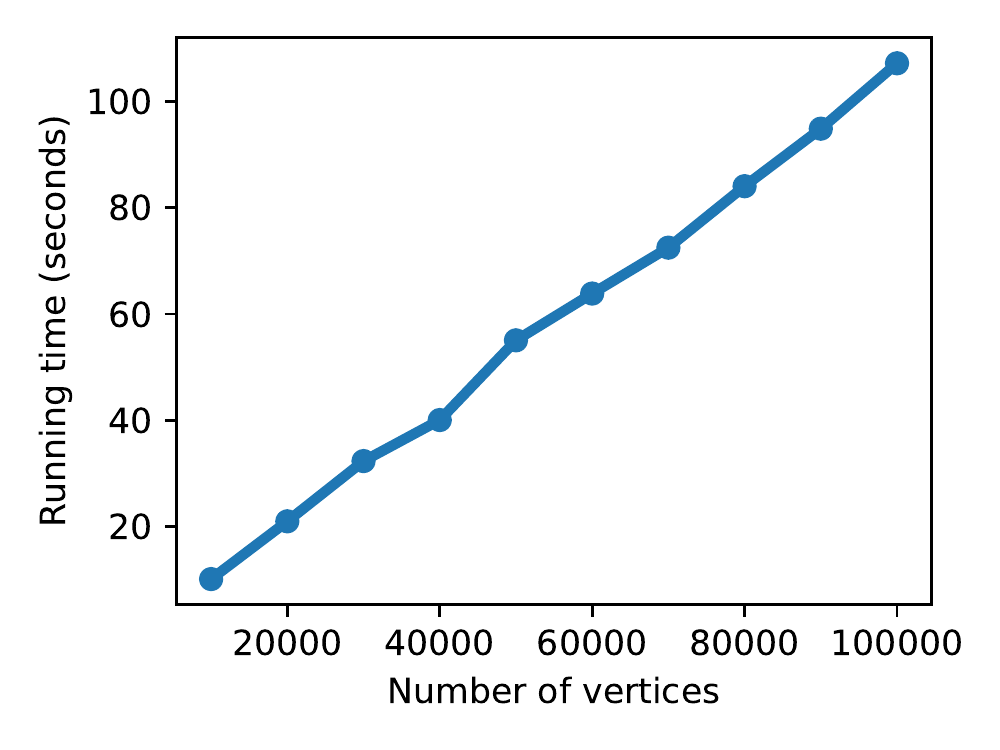}%
\includegraphics[width=0.5\textwidth]{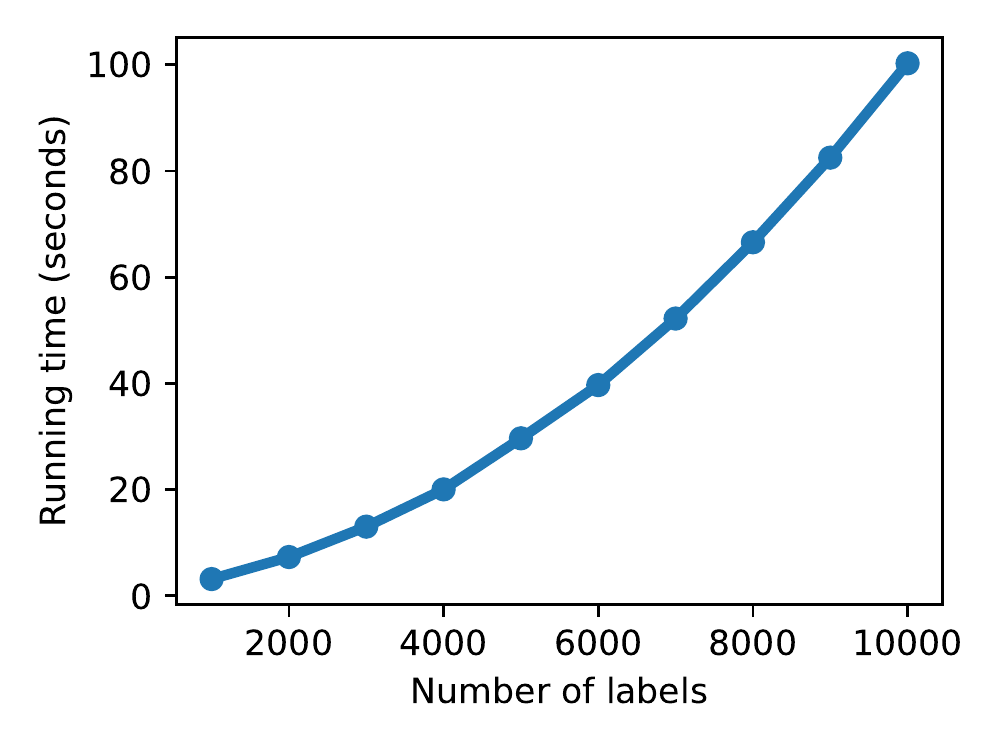}
}
\caption{Running time of the \alggreedyor algorithm as a function of the number of vertices (left) and the number of labels (right) in our synthetic graphs.}
\label{runtime_disjunction}
\end{figure}

We repeat the experiments with increasing values of $\epsilon$ and compare the density of the subgraph induced by the target labels to the density of the subgraph induced by the labels of the greedy algorithms. For each $\epsilon$, we run the experiment 10 times and compute the mean and standard deviation of the runs. The results are shown in Figure~\ref{synthetic}.

In both cases, the greedy algorithms correctly find the target labels for small values of $\epsilon$. After $\epsilon > 0.25$ for \alggreedyand and after $\epsilon > 0.35$ for \alggreedyor, the algorithms start to find other sets of labels, which yield higher densities than the target labels as many of the edges in the target clique have been removed and other edges have been added. However, at $\epsilon = 0.30$, the \alggreedyor returns a suboptimal solution that yields a slightly lower density than the target labels.

We confirm the theoretical running times of the algorithms by setting $\epsilon=0.2$ and performing experiments with increasingly large graphs, where the number of total vertices goes from 10000 up to 100000 while other aspects of the experiments remain constant. Similarly, we test how the running times of the algorithms scale as the number of total labels in our synthetic graph increases from 1000 to 10000. The results for \alggreedyand are shown in Figure~\ref{runtime_conjunction} and results for \alggreedyor in Figure~\ref{runtime_disjunction}.

As expected, the running times of both algorithms scale linearly with the number of vertices in the graph. Furthermore, the running time of our naive implementation of \alggreedyor appears to scale quadratically with the number of labels, while the scaling for \alggreedyand is close to linear. These results confirm our theoretical analysis and show that our algorithms can be applied to large graphs in practice.

\textbf{Experiments with real-world datasets:}
We test the greedy algorithms by running experiments on four real-world datasets. The first dataset is the Enron Email Dataset\footnote{\url{https://www.cs.cmu.edu/~./enron/}}, which consists of publicly available emails from employees of a former company called Enron Corporation. We collect the emails in sent mail folders and construct a graph where new edges are added between the sender and the recipients of each email. Each edge has labels consisting of the stemmed words in the email's title, with stop words and words including numbers removed. 

The second dataset consists of high energy physics theory publications (HEP-TH) from the years 1992 to 2003. The data was originally released in KDD Cup\footnote{\url{https://www.cs.cornell.edu/projects/kddcup/datasets.html}} but we use a preprocessed version of the data available in GitHub\footnote{\url{https://github.com/chriskal96/physics-theory-citation-network}} We create the network by adding authors as vertices, and edges between any two authors are added if they share at least two publications. The edges between authors are then given labels which consist of the stemmed words in the titles of the shared articles between the two authors. We exclude stop words and words including numbers from the titles the same way as for the Enron dataset.

The third dataset consists of publications from the DBLP\footnote{\url{https://www.aminer.org/citation}} dataset~\citep{Tang:08KDD}. From this dataset, we chose publications from ECMLPKDD, ICDM, KDD, NIPS, SDM, and WWW conferences. The network is constructed in the same way as for the HEP-TH data, with authors as vertices, two or more shared publications as edges, and stemmed and filtered words from the titles as labels.

The fourth and final dataset consists of the latest 10000 tweets collected from Twitter API\footnote{\url{https://developer.twitter.com/en/docs/twitter-api}} with the hashtag \#metoo by the 27th of May, 23:59 UTC. We create the network by having users as vertices with an edge between a pair of users if one of them has retweeted or responded to one of the other's tweets. The labels on the edge are then any hashtags in the retweets or response tweets between the two users.

We construct the networks by filtering out labels that appear in less than 0.1\% of the edges in the Enron and Twitter datasets, or labels that occur in less than 0.5\% of the papers in the case of the HEP-TH and DBLP datasets. The sizes, label counts, and densities of the resulting graphs are shown in Table~\ref{graph_sizes}. 

We run the greedy algorithms on each of these graphs, and compare the results against the densest subgraph ignoring the labels (\algbase).
 We report the statistics for the label-induced subgraphs and the densest subgraphs in Table~\ref{subgraph_results}. 

For each of the datasets, both algorithms find label-induced subgraphs with higher densities than in the original graphs. In most cases, the restriction of constructing label-induced subgraphs results in clearly lower densities compared to the densest label-ignorant subgraphs. Interestingly, for the DBLP dataset \alggreedyand finds a label-induced subgraph with a very high density that is close to the density of the densest subgraph ignoring the labels. The running times are practical: the algorithm processes networks with 100 000 edge-label pairs in seconds.

For Enron and HEP-TH datasets, the \alggreedyor returns large sets of labels resulting in large subgraphs, whereas the \alggreedyand algorithm selects only a few labels with smaller induced subgraphs in each case. For the Twitter dataset, both greedy algorithms select only one label, which induces a small subgraph with a notably higher density than the original graph. 

\begin{table}[t!]

\caption{Basic characteristics of the networks: number or vertices $\abs{V}$, number or edges $\abs{E}$, number of labels $\abs{L}$, number of edge-label pairs $p$, and the density $d(G) = \abs{E} / \abs{V}$.}

\begin{filecontents*}{graph_results.txt}
dataset,n,m,L,p,density,t-con,n-con,m-con,density-con,L-con,t-dis,n-dis,m-dis,density-dis,L-dis,density-decomp,n-decomp
\dtname{Enron},11024,18072,2604,361000,1.63933236574746,10.433538,18,31,1.7222222222222223,2,25.022379,1233,2711,2.1987023519870235,193,11.352941,85
\dtname{HEP-TH},4738,7767,240,78078,1.6392992823976362,1.962063,7,14,2.0,4,5.741033,3284,5588,1.7015834348355663,40,3.810345,58
\dtname{DBLP},10550,16811,268,160850,1.5934597156398105,4.058819,25,300,12.0,3,1.742327,243,538,2.213991769547325,1,12.522727,44
\dtname{Twitter},7973,9314,248,19849,1.1681926501944062,0.770946,12,31,2.5833333333333335,1,1.885532,12,31,2.5833333333333335,1,3.368421,19
\end{filecontents*}
\pgfplotstabletypeset[
    begin table=\begin{tabular*}{\textwidth},
    end table=\end{tabular*},
    col sep=comma,
    columns = {dataset,n,m,L,p,density},
    columns/dataset/.style={string type, column type={@{\extracolsep{\fill}}l}, column name=\emph{Dataset}},
    columns/n/.style={fixed, set thousands separator={\,}, column type=r, column name=$\abs{V}$},
    columns/m/.style={fixed, set thousands separator={\,}, column type=r, column name=$\abs{E}$},
    columns/L/.style={fixed, set thousands separator={\,}, column type=r, column name=$\abs{L}$},
    columns/p/.style={fixed, set thousands separator={\,}, column type=r, column name=$p$},
    columns/density/.style={fixed, set thousands separator={\,}, column type=r, column name=$d(G)$},
    every head row/.style={
        before row={\toprule},
            after row=\midrule},
    every last row/.style={after row=\bottomrule},
]
{graph_results.txt}
\label{graph_sizes}
\end{table}

\begin{table}[t!]

\caption{Statistics for the resulting subgraphs for the greedy algorithms and the label-ignorant densest subgraph algorithm. For the label-induced subgraphs, we have the number of vertices $n$, the number of edges $m$, the size of the best set of labels $\abs{A}$, density $d$, and running time $t$ in seconds. For the densest subgraph, we show the number of vertices $n$ and density $d = m/n$.}
\setlength{\tabcolsep}{0pt}

\pgfplotstabletypeset[
    begin table=\begin{tabular*}{\textwidth},
    end table=\end{tabular*},
    col sep=comma,
    columns = {dataset,n-con,m-con,L-con,density-con,t-con,n-dis,m-dis,L-dis,density-dis,t-dis,n-decomp,density-decomp},
    columns/dataset/.style={string type, column type={@{\extracolsep{\fill}}l}, column name=\emph{Dataset}},
    columns/density/.style={fixed, set thousands separator={\,}, column type=r, column name=$d$},
    columns/n-con/.style={fixed, set thousands separator={\,}, column type=r, column name=$n$},
    columns/m-con/.style={fixed, set thousands separator={\,}, column type=r, column name=$m$},
    columns/density-con/.style={fixed, set thousands separator={\,}, column type=r, column name=$d$},
    columns/L-con/.style={fixed, set thousands separator={\,}, column type=r, column name=$\abs{A}$},
    columns/t-con/.style={fixed, set thousands separator={\,}, column type=r, column name=$t$},
    columns/n-dis/.style={fixed, set thousands separator={\,}, column type=r, column name=$n$},
    columns/m-dis/.style={fixed, set thousands separator={\,}, column type=r, column name=$m$},
    columns/density-dis/.style={fixed, set thousands separator={\,}, column type=r, column name=$d$},
    columns/L-dis/.style={fixed, set thousands separator={\,}, column type=r, column name=$\abs{A}$},
    columns/t-dis/.style={fixed, set thousands separator={\,}, column type=r, column name=$t$},
    columns/density-decomp/.style={fixed, set thousands separator={\,}, column type=r, column name=$d$},
    columns/n-decomp/.style={fixed, set thousands separator={\,}, column type=r, column name=$n$},
    every head row/.style={
        before row={\toprule
        & \multicolumn{5}{l}{\alggreedyand}
        & \multicolumn{5}{l}{\alggreedyor}
        & \multicolumn{2}{l}{\algbase}\\
        \cmidrule(r){2-6}
        \cmidrule(r){7-11}
        \cmidrule(r){12-13}
        },
            after row=\midrule},
    every last row/.style={after row=\bottomrule},
]
{graph_results.txt}
\label{subgraph_results}
\end{table}

\textbf{Experiments with $\alpha$-density:} Next we consider finding $\alpha$-dense subgraphs by running the \alggreedyandalpha and \alggreedyoralpha algorithms on the same datasets. The results are shown in Table~\ref{alpha_conjunction_results} and Table~\ref{alpha_disjunction_results}, respectively.

As pointed out by Theorem~\ref{thr:alphamax}, the optimal $\alpha$-dense subgraph is also the densest for sufficiently large $\alpha$. We use a binary search to find the maximum $\alpha$ for which the greedy algorithm yields a non-empty graph. The values of $\alpha$ in these tables are chosen by the binary search process while searching for the maximum. Additionally, we experiment with using a smaller $\alpha$ value of 0.25 times the maximum. For clarity, we exclude duplicated results where different values of $\alpha$ yield the same subgraph.

We see that the greedy algorithms for the two problems often find the same solution, as suggested by Theorem~\ref{thr:alphamax}. However, this is not always the case due to the heuristic nature of these algorithms. Interestingly, with $\alpha=2.5$ for the HEP-TH dataset, the \alggreedyandalpha finds a denser subgraph than the one found by \alggreedyand, while an additional manual experiment with $\alpha=1.4$ results in the greedy algorithm suboptimally returning an empty graph. For the DBLP dataset using $\alpha = 3.6$ leads to the same solution as \alggreedyand, but larger values of $\alpha$ lead the greedy algorithm to choose a suboptimal first label resulting in less dense subgraphs.
For Enron and HEP-TH datasets, \alggreedyoralpha only finds subgraphs with a slightly lower density than the ones found by \alggreedyor. 

In general, we observe that using smaller values of $\alpha$ results in subgraphs with more vertices and edges in both the disjunctive and conjunctive cases. Thus having $\alpha$ as a parameter gives us more control over the size of the resulting subgraph and allows us to look for both smaller and larger groups of densely connected nodes.

\begin{table}[t!]
\caption{Results for running \alggreedyandalpha on the four datasets with the different values of $\alpha$. For each resulting subgraph, we have the density $d = m/n$, the chosen labels, the number of nodes $n$, and the number of edges $m$. Results matching the densest subgraph found by \alggreedyand are shown in bold.} 
\hfill
\setlength{\tabcolsep}{1mm}
\adjustbox{valign=t}{
\begin{tabular}{llllll}
\toprule
Dataset & $\alpha$ & $d$ & labels & $n$ & $m$ \\
\midrule
Enron & 1.032 & 1.5088 & meet & 1875 & 2829 \\
 & 1.548 & 1.5971 & legal & 479 & 765 \\
 & \textbf{1.7218} & \textbf{1.7222} & \textbf{mopa, action} & \textbf{18} & \textbf{31} \\[1mm]
HEP-TH & 0.6249 & 1.3656 & theori & 2410 & 3291 \\
 & 2.4998 & 2.5 & casimir, light & 6 & 15 \\[1mm]
DBLP & 1.0927 & 1.3005 & learn & 3780 & 4916 \\
 & \textbf{3.6} & \textbf{12} & \textbf{novel, rate, techniqu} & \textbf{25} & \textbf{300} \\
 & 4.3708 & 6.2 & forecast, experi, use & 15 & 93 \\[1mm]
Twitter & 0.6457 & 1.1381 & metoo & 8297 & 7290 \\
 & \textbf{2.5828} & \textbf{2.5833} & \textbf{metooase} & \textbf{12} & \textbf{31} \\
\bottomrule
\end{tabular}}
\hspace*{\fill}
\label{alpha_conjunction_results}
\end{table}

\begin{table}[t!]
\caption{Results for running \alggreedyoralpha on the four datasets with the different values of $\alpha$. For each resulting subgraph, we have the density $d=m/n$, the chosen labels or their amount, the number of nodes $n$, and the number of edges $m$. Results matching the densest subgraph found by \alggreedyor are shown in bold.} 
\hfill
\adjustbox{valign=t}{
\begin{tabular*}{\textwidth}{@{\extracolsep{\fill}}lllp{4cm}ll}
\toprule
Dataset & $\alpha$ & $d$ & labels & $n$ & $m$ \\
\midrule
Enron & 1.32 & 1.8497 & (553 labels) & 7982 & 14765 \\
 & 1.98 & 2.1715 & (964 labels) & 3633 & 7889 \\
 & 2.145 & 2.1810 & (577 labels) & 2497 & 5446 \\
 & 2.1656 & 2.1798 & (547 labels) & 2364 & 5153 \\
 & 2.1798 & 2.1799 & (582 labels) & 2329 & 5077 \\[1mm]
HEP-TH & 0.4255 & 1.6393 & (75 labels) & 4738 & 7767 \\
 & 1.02 & 1.6557 & (68 labels) & 4650 & 7699 \\
 & 1.53 & 1.6938 & (76 labels) & 4063 & 6882 \\
 & 1.6575 & 1.7011 & (52 labels) & 3567 & 6068 \\
 & 1.7013 & 1.7023 & (50 labels) & 3426 & 5832 \\[1mm]
DBLP & 0.5534 & 1.5953 & (93 labels) & 10532 & 16802 \\
 & 1.326 & 1.6098 & (83 labels) & 10295 & 16573 \\
 & \textbf{2.2137} & \textbf{2.2140} & \textbf{novel} & \textbf{243} & \textbf{538} \\[1mm]
Twitter & 0.6457 & 1.1685 & (49 labels) & 7969 & 9312 \\
 & 1.548 & 1.8857 & metooase, streamer, anubhavmohanty, victimservices, causette, istandwithjohnny & 70 & 132 \\
 & \textbf{2.5828} & \textbf{2.5833} & \textbf{metooase} & \textbf{12} & \textbf{31} \\
\bottomrule
\end{tabular*}}
\hspace*{\fill}
\label{alpha_disjunction_results}
\end{table}

\textbf{Case study:} We analyze the label-induced dense subgraphs for the Twitter and DBLP datasets by repeatedly running the \alggreedyand algorithm for these graphs. After running the algorithm, we exclude the edges from the output edge-induced subgraph and run the algorithm again on the remaining graph.  The first 8 resulting sets of labels, as well as densities and sizes for the induced subgraphs, are shown in Table~\ref{labels_results}.

For the DBLP graph, the algorithm finds a group of 25 authors that have each written at least two papers together with a shared topic, as well as other relatively large groups of authors whose edges form almost perfect cliques. The labels representing stemmed words can be used to interpret the topics of publications for these groups of authors having tight collaboration.

For the Twitter data of \#metoo tweets, the densest label-induced subgraphs are formed by mostly looking at individual hashtags. This detects groups of people tweeting about \#MeTooASE referring to the French Me Too movement for foster children, as well as groups closely discussing other topics in the context of the Me Too movement such as live streaming or the recent trial between Johnny Depp and Amber Heard.

We see that the same labels also appear when searching for $\alpha$-dense subgraphs.
For example, by looking at the labels for $\alpha=1.548$ for the Twitter dataset in Table~\ref{alpha_disjunction_results} and comparing them with the labels in Table~\ref{labels_results}, we can see that this subgraph found by the \alggreedyoralpha algorithm in fact consists of multiple smaller groups of people discussing a variety of topics that we previously discovered.

\begin{table}[t!]
\caption{Label sets with corresponding subgraph densities and sizes selected by running the \alggreedyand algorithm repeatedly on the graphs for DBLP and Twitter datasets. The labels are stemmed words from publication titles for DBLP, and tweet hashtags for Twitter data. The densities are not monotonically decreasing as the greedy algorithm does not always find the optimal solution.}

\hfill
\setlength{\tabcolsep}{1mm}
\adjustbox{valign=t}{\begin{tabular}{lp{4cm}ll}

& DBLP \\
\toprule
$d$ & labels & $n$ & $m$ \\
\midrule
12.0 & novel, rate, techniqu & 25 & 300 \\
10.74 & identif, combin, process & 23 & 247 \\
6.2 & forecast, experi, use & 15 & 93 \\
6.0 & heterogen, manag, stream, use & 13 & 78 \\
2.0 & heterogen, segment & 5 & 10 \\
3.13 & heterogen, manag, use, dynam & 8 & 25 \\
2.5 & heterogen, sourc, toward & 6 & 15 \\
2.5 & heterogen, construct, dimension, network & 6 & 15 \\
\bottomrule
\end{tabular}}
\hfill
\adjustbox{valign=t}{
\begin{tabular}{llll}
& Twitter \\
\toprule
$d$ & labels & $n$ & $m$ \\
\midrule
2.58 & metooase & 12 & 31 \\
1.88 & streamer & 16 & 30 \\
1.75 & anubhavmohanty & 16 & 28 \\
1.71 & victimservices & 7 & 12 \\
1.83 & causette, lfi & 6 & 11 \\
1.63 & istandwithjohnny & 8 & 13 \\
1.43 & rupertmurdock & 7 & 10 \\
1.25 & marilynmanson & 8 & 10 \\
\bottomrule
\end{tabular}}
\hspace*{\fill}
\label{labels_results}
\end{table}
\section{Concluding remarks}\label{sec:conclusions}

In this paper, we considered the problem of finding dense subgraphs that are induced by labels on the edges. More specifically, we considered two cases: conjunctive-induced dense subgraphs, where the edges need to contain the given label set, and disjunctive-induced dense subgraphs,
where the edges need to have only one label in common. As a measure of quality, we used the average degree of a subgraph.
We showed that both problems are \np-hard, and we proposed a greedy heuristic to find dense induced subgraphs. By maintaining suitable counters we were able to find subgraphs in quasi-linear time: $\bigO{p \log \abs{L}}$
for conjunctive-induced graphs and $\bigO{p \log^2 \abs{L}}$ for disjunctive-induced graphs.
In addition, we analyzed the related problem of maximizing the number of edges minus $\alpha$ times the number of vertices and showed how the optimal solutions to these problems are connected. We proved that the problem of maximizing this $\alpha$-density is \np-hard and inapproximable unless $\poly=\np$. We adopted the greedy algorithms for the conjunctive and disjunctive cases of this problem resulting in a running time of $\bigO{p \log \abs{L}}$ for the disjunctive case as well.
We then demonstrated that the algorithms are practical, they can find ground truth in synthetic datasets, and find interpretable results from real-world networks.

While this paper focused on the conjunctive and disjunctive cases, future work could explore other ways to induce graphs from a label set and design efficient algorithms for such tasks. Another direction for future work is to relax the requirement that every edge/node must be induced from labels. Instead, we can allow some deviation from this requirement but then penalize the deviations appropriately when assessing the quality of the subgraph.

\section{Declarations}

\textbf{Funding}
This research is supported by the Academy of Finland projects MALSOME (343045).

\noindent\textbf{Disclosure of potential conflicts of interest}
Not applicable

\noindent\textbf{Ethics approval}
Not applicable

\noindent\textbf{Consent to participate}
Not applicable

\noindent\textbf{Consent for publication}
Not applicable

\noindent\textbf{Availability of data and material}
Publicly available datasets were used. See Section~\ref{sec:exp} for the links.

\noindent\textbf{Code availability}
\url{https://version.helsinki.fi/dacs/}.

\noindent\textbf{Authors' contributions}
Nikolaj Tatti formulated the problems. Iiro Kumpulainen implemented the algorithms and conducted the experiments. Both authors wrote the manuscript.

\noindent\textbf{Version of Record}
This version of the article has been accepted for publication, after peer review but is not the Version of Record and does not reflect post-acceptance improvements, or any
corrections. The Version of Record is available online at: \url{http://dx.doi.org/10.1007/s10994-023-06377-y}

\bibliography{references}

\begin{thebibliography}{25}
\providecommand{\natexlab}[1]{#1}
\providecommand{\url}[1]{\texttt{#1}}
\expandafter\ifx\csname urlstyle\endcsname\relax
  \providecommand{\doi}[1]{doi: #1}\else
  \providecommand{\doi}{doi: \begingroup \urlstyle{rm}\Url}\fi

\bibitem[Abello et~al.(2002)Abello, Resende, and Sudarsky]{abello02clique}
James Abello, G.C. Resende, Mauricio, and Sandra Sudarsky.
\newblock Massive quasi-clique detection.
\newblock In \emph{LATIN 2002: Theoretical Informatics}, pages 598--612, 2002.

\bibitem[Angel et~al.(2014)Angel, Koudas, Sarkas, Srivastava, Svendsen, and Tirthapura]{angel2014dense}
Albert Angel, Nick Koudas, Nikos Sarkas, Divesh Srivastava, Michael Svendsen, and Srikanta Tirthapura.
\newblock Dense subgraph maintenance under streaming edge weight updates for real-time story identification.
\newblock \emph{The VLDB journal}, 23\penalty0 (2):\penalty0 175--199, 2014.

\bibitem[Balasundaram et~al.(2011)Balasundaram, Butenko, and Hicks]{balasundaram:2011:kflex}
Balabhaskar Balasundaram, Sergiy Butenko, and Illya~V. Hicks.
\newblock Clique relaxations in social network analysis: The maximum $k$-plex problem.
\newblock \emph{Operations Research}, 59\penalty0 (1):\penalty0 133--142, 2011.

\bibitem[Bonchi et~al.(2019)Bonchi, Khan, and Severini]{bonchi2019distance}
Francesco Bonchi, Arijit Khan, and Lorenzo Severini.
\newblock Distance-generalized core decomposition.
\newblock In \emph{SIGMOD}, pages 1006--1023, 2019.

\bibitem[Brodal and Jacob(2002)]{brodal2002dynamic}
Gerth~St{\o}lting Brodal and Riko Jacob.
\newblock Dynamic planar convex hull.
\newblock In \emph{FOCS}, pages 617--626, 2002.

\bibitem[Bron and Kerbosch(1973)]{Bron:1973:AFC:362342.362367}
Coen Bron and Joep Kerbosch.
\newblock Algorithm 457: Finding all cliques of an undirected graph.
\newblock \emph{Communications of the ACM}, 16\penalty0 (9):\penalty0 575--577, 1973.

\bibitem[Charikar(2000)]{Charikar:2000tg}
Moses Charikar.
\newblock {Greedy approximation algorithms for finding dense components in a graph}.
\newblock \emph{APPROX}, 2000.

\bibitem[Danisch et~al.(2017)Danisch, Chan, and Sozio]{danisch2017large}
Maximilien Danisch, T-H~Hubert Chan, and Mauro Sozio.
\newblock Large scale density-friendly graph decomposition via convex programming.
\newblock In \emph{Proceedings of the 26th International Conference on World Wide Web}, pages 233--242. International World Wide Web Conferences Steering Committee, 2017.

\bibitem[Dinkelbach(1967)]{dinkelbach1967nonlinear}
Werner Dinkelbach.
\newblock On nonlinear fractional programming.
\newblock \emph{Management science}, 13\penalty0 (7):\penalty0 492--498, 1967.

\bibitem[Du et~al.(2009)Du, Jin, Ding, Lee, and Thornton~Jr]{du2009migration}
Xiaoxi Du, Ruoming Jin, Liang Ding, Victor~E Lee, and John~H Thornton~Jr.
\newblock Migration motif: a spatial-temporal pattern mining approach for financial markets.
\newblock In \emph{KDD}, pages 1135--1144, 2009.

\bibitem[Fratkin et~al.(2006)Fratkin, Naughton, Brutlag, and Batzoglou]{fratkin2006motifcut}
Eugene Fratkin, Brian~T Naughton, Douglas~L Brutlag, and Serafim Batzoglou.
\newblock Motifcut: regulatory motifs finding with maximum density subgraphs.
\newblock \emph{Bioinformatics}, 22\penalty0 (14):\penalty0 e150--e157, 2006.

\bibitem[Galbrun et~al.(2014)Galbrun, Gionis, and Tatti]{galbrun2014overlapping}
Esther Galbrun, Aristides Gionis, and Nikolaj Tatti.
\newblock Overlapping community detection in labeled graphs.
\newblock \emph{DMKD}, 28\penalty0 (5):\penalty0 1586--1610, 2014.

\bibitem[Goldberg(1984)]{Goldberg:1984up}
Andrew~V Goldberg.
\newblock {Finding a maximum density subgraph}.
\newblock \emph{University of California Berkeley Technical report}, 1984.

\bibitem[H{\aa}stad(1996)]{DBLP:conf/focs/Hastad96}
Johan H{\aa}stad.
\newblock Clique is hard to approximate within $n^{1 - \epsilon}$.
\newblock In \emph{FOCS}, pages 627--636, 1996.

\bibitem[Kumpulainen and Tatti(2022)]{kumpulainen2022community}
Iiro Kumpulainen and Nikolaj Tatti.
\newblock Community detection in edge-labeled graphs.
\newblock In \emph{Discovery Science: 25th International Conference, DS 2022, Montpellier, France, October 10--12, 2022, Proceedings}, pages 460--475, 2022.

\bibitem[Langston et~al.(2005)Langston, Lin, Peng, Baldwin, Symons, Zhang, and Snoddy]{langston2005combinatorial}
Michael~A Langston, Lan Lin, Xinxia Peng, Nicole~E Baldwin, Christopher~T Symons, Bing Zhang, and Jay~R Snoddy.
\newblock A combinatorial approach to the analysis of differential gene expression data.
\newblock In \emph{Methods of Microarray Data Analysis}, pages 223--238. Springer, 2005.

\bibitem[Li and Klette(2011)]{Li2011}
Fajie Li and Reinhard Klette.
\newblock \emph{Euclidean Shortest Paths: Exact or Approximate Algorithms}, chapter Convex Hulls in the Plane, pages 93--125.
\newblock Springer, 2011.

\bibitem[Mokken(1979)]{mokken1979cliques}
Robert~J Mokken.
\newblock Cliques, clubs and clans.
\newblock \emph{Quality \& Quantity}, 13\penalty0 (2):\penalty0 161--173, 1979.

\bibitem[Overmars and Van~Leeuwen(1981)]{overmars1981maintenance}
Mark~H Overmars and Jan Van~Leeuwen.
\newblock Maintenance of configurations in the plane.
\newblock \emph{Journal of computer and System Sciences}, 23\penalty0 (2):\penalty0 166--204, 1981.

\bibitem[Pool et~al.(2014)Pool, Bonchi, and Leeuwen]{pool2014description}
Simon Pool, Francesco Bonchi, and Matthijs~van Leeuwen.
\newblock Description-driven community detection.
\newblock \emph{TIST}, 5\penalty0 (2):\penalty0 1--28, 2014.

\bibitem[Seidman(1983)]{seidman1983network}
Stephen~B Seidman.
\newblock Network structure and minimum degree.
\newblock \emph{Social networks}, 5\penalty0 (3):\penalty0 269--287, 1983.

\bibitem[Tang et~al.(2008)Tang, Zhang, Yao, Li, Zhang, and Su]{Tang:08KDD}
Jie Tang, Jing Zhang, Limin Yao, Juanzi Li, Li~Zhang, and Zhong Su.
\newblock Arnetminer: Extraction and mining of academic social networks.
\newblock In \emph{KDD}, pages 990--998, 2008.

\bibitem[Tatti(2019)]{tatti2019density}
Nikolaj Tatti.
\newblock Density-friendly graph decomposition.
\newblock \emph{TKDD}, 13\penalty0 (5):\penalty0 1--29, 2019.

\bibitem[Tsourakakis(2015)]{tsourakakis15triangle}
Charalampos~E. Tsourakakis.
\newblock The k-clique densest subgraph problem.
\newblock In \emph{WWW}, pages 1122--1132, 2015.

\bibitem[Uno(2010)]{Uno:2010:EAS:1712671.1712672}
Takeaki Uno.
\newblock An efficient algorithm for solving pseudo clique enumeration problem.
\newblock \emph{Algorithmica}, 56\penalty0 (1):\penalty0 3--16, 2010.

\end{thebibliography}

\appendix

\section{Computational complexity proofs}\label{sec:proofs}

\begin{proof}[Proof of Theorem~\ref{thr:npand}]
We will prove the claim by reducing \textsc{3ExactCover} to the densest subgraph problem. In \textsc{3ExactCover} we are given a set $X$ and a family $\fm{C}$ of subsets of size 3 over $X$ and asked if there is a disjoint subset of $\fm{C}$ whose union is $X$.

Assume that we are given a set $X$ and a family $\fm{C} = \set{C_1, \ldots, C_N}$ of $N$ subsets. We set labels to be $L = \enset{1}{N}$. The vertices $V$ contain $N$ vertices $y_1, \ldots, y_N$, and an additional vertex $z$. We connect each $y_i$ to $z$, labeled with $L \setminus \set{i}$. For each overlapping $C_i$ and $C_j$, we introduce
$4N$ additional vertices and
$2N$ edges, each edge connecting two \emph{unique nodes}, and labeled as 
$L \setminus \set{i, j}$.

We claim that for $\abs{X} \geq 5$,
\textsc{3ExactCover} has a solution if and only if
there is an induced graph $H$ with $d(H) \geq \abs{X}/ (\abs{X} + 3)$. 

Assume that we are given a set of labels $A \subset L$. Let $B = L \setminus A$.
Let $k$ be the number of set pairs in $B$ that are overlapping, that is,
\[
    k = \abs{\set{\set{i, j} \mid i, j \in B, C_i \cap C_j \neq \emptyset}}.
\]

Then the density of the corresponding graph $H = G(\indand{}, A)$ is equal to
\[
    d(H) = \frac{\abs{B} + 2Nk}{\abs{B} + 1 + 4Nk}\quad.
\]

Assume that $k > 0$. Since $\abs{B} \leq N$, we can bound the density with
\[
    d(H) = \frac{\abs{B} + 2Nk}{\abs{B} + 1 + 4Nk}
    \leq \frac{N + 2Nk}{N + 1 + 4Nk}
    <  \frac{N + 2Nk}{N + 4Nk} \leq \frac{3}{5}\quad.
\]

Assume that $k = 0$. Then the density is equal to $\abs{B}/ (\abs{B} + 1)$.
Let $\fm{U} = \set{C_i \mid i \in B}$. Since $\fm{U}$ is disjoint, $3\abs{B} \leq \abs{X}$ and the equality holds if and only if $\fm{U}$ covers $X$.

Assume that there is a subgraph $H = G(\indand{}, A)$ with
$d(H) \geq \abs{X}/ (\abs{X} + 3)$.
Since we assume that $\abs{X} \geq 5$, we have $d(H) \geq 5/8 > 3/5$, and the preceding discussion shows that the sets
corresponding to $A$ form
an exact cover of $X$.

On the other hand, if there is an exact cover in $\fm{C}$, then $d(G(\indand{}, A)) = \abs{X}/ (\abs{X} + 3)$, where $A$ is the set of labels corresponding to the cover. This shows that maximizing the density of the label-induced subgraph is an \np-hard problem.
\end{proof}

\begin{proof}[Proof of Theorem~\ref{thr:npor}]
We will prove the claim by reducing \textsc{3ExactCover} to the densest subgraph problem. In \textsc{3ExactCover} we are given a set $X$ and a family $\mathcal{C}$ of subsets of size 3 over $X$ and asked if there is a disjoint subset of $\mathcal{C}$ whose union is $X$.

Assume that we are given a set $X$ and a family $\fm{C} = \set{C_1, \ldots, C_N}$ of $N$ subsets. The vertices $V$ consists of the set $X$, $N$ additional vertices $y_1, \ldots, y_N$, and 2 more vertices $Z = z_1, z_2$. We have $N$ labels, $L = \enset{1}{N}$.

Next, we define the edges $E$. Connect each $x \in X$ to $Z$,
and label the edges with labels $\set{i \mid x \in C_i}$. Then for each $C_i$, we connect $z_1$ to $y_i$, labeled with $i$.

We claim that \textsc{3ExactCover} has a solution if and only if
the optimal label-induced graph has the density of $7\abs{X} / (6 + 4 \abs{X})$.

Given a non-empty set of labels $A \subseteq L$, the density of the corresponding graph $H$ is equal to $g(k, \abs{A})$, where
    $g(s, t) = \frac{2s + t}{2 + s + t}$,
and $k$ is the size of the union of sets in $\fm{C}$ corresponding to $A$.

Note that since $k \geq 3$, we have $2k > 2 + k$. Thus, $\partial \log g / \partial t = 1/(2k + t) - 1/(2 + k + t) < 0$, and consequently $g(k, t) > g(k, t')$ when $t < t'$.

Since each set in $\fm{C}$ is of size 3, we have $\abs{A} \geq k/3$. Thus,
\[
    g(k, \abs{A}) \leq g(k, k/3) = \frac{7k}{6 + 4k} \leq \frac{7\abs{X}}{6 + 4\abs{X}},
\]
where the equalities hold if and only if $k = \abs{X}$ and $3\abs{A} = k$,
that is, $A$ corresponds to an exact cover of $X$.
\end{proof}

\begin{proof}[Proof of Corollary~\ref{cor:npalpha}]
Let us adopt the notation of the proof of Theorem~\ref{thr:npor}.
The proof shows that \textsc{3ExactCover} has a solution if and only if
there is an induced graph $H$ with $7d(H) \geq \abs{X}/ (4\abs{X} + 6)$.
Moreover, there are $N +2$ nodes in the graph, so a difference between two densities is at least $(N + 2)^{-2}$.
Consequently,
if we set $\tau = 7d(H) \geq \abs{X}/ (4\abs{X} + 6) - 0.5(N + 2)^{-2}$, then Theorem~\ref{thr:alphamax} implies that
\textsc{3ExactCover} has a solution if and only if there is $H$ with $g(H, \tau) > 0$. This proves the hardness and the inapproximability since any algorithm with a multiplicative guarantee will find the optimal solution.
The proof for $\indand{}$ is similar.
\end{proof}

\end{document}